\documentclass[12pt]{article}

\usepackage{algorithm}
\usepackage{algorithmic}
\usepackage{amsmath}
\usepackage{amssymb}
\usepackage{amsthm}
\usepackage{enumerate}
\usepackage{graphicx}
\usepackage{listings}

\newcommand{\complain}[1]{}

\usepackage{color} 

\def\Reals{\ensuremath{\mathbb{R}}}

\newtheorem{definition}{Definition}
\newtheorem{lemma}[definition]{Lemma}
\newtheorem{proposition}[definition]{Proposition}
\newtheorem{theorem}[definition]{Theorem}

\newcommand{\eps}{\varepsilon}
\newcommand{\cclass}[1]{{#1}}
\newcommand{\problem}[1]{\textsc{#1}}
\newcommand{\co}{\colon\thinspace}

\newcommand{\fpt}{\cclass{FPT}}
\newcommand{\wone}{\cclass{W[1]}}

\newcommand{\clique}{\problem{Clique}}

\newcommand{\ap}{a.p.\ }
\renewcommand{\ap}{axis-parallel\ }
\newcommand{\wrt}{w.r.t.\ }
\renewcommand{\wrt}{with respect to\ }

\newcommand{\R}{\mathbb{R}}

\newcommand{\Rtwo}{\mathbb{R}^2}

\newcommand{\Rd}{\mathbb{R}^d}
\newcommand{\N}{\mathbb{N}}

\newcommand{\order}{\mathcal{O}}

\newcommand{\calS}{\mathcal{S}}

\newcommand{\calD}{\mathcal{D}}
\newcommand{\calR}{\mathcal{R}}

\newcommand{\calW}{\mathcal{W}}

\newcommand{\calU}{\mathcal{U}}

\newcommand{\pr}{\operatorname{pr}}
\newcommand{\dist}{\operatorname{dist}}

\newcommand{\rep}{\operatorname{rep}}
\newcommand{\BB}{\operatorname{BB}}
\newcommand{\diam}{\operatorname{diam}}

\newcommand{\follows}{\Rightarrow}

\newcommand{\rup}[1]{ \left\lceil  #1 \right\rceil}

\title{Fixed-parameter tractability and lower bounds\\ for stabbing problems}
\author{ Panos Giannopoulos%
  \thanks{Institut f{\"ur} Informatik, Freie Universit{\"a}t Berlin,
    Takustra{\ss}e 9, D-14195 Berlin, Germany, {\tt\{panos, knauer, rote,
    dwerner\}@inf.fu-berlin.de}.}{\ }\footnote{This research was
    supported by the German Science Foundation (DFG) under grant Kn~591/3-1.}
  \and
  Christian Knauer\footnotemark[1]{\ }\footnotemark[2]
  \and
  G{\"u}nter Rote\footnotemark[1]
  \and
  Daniel Werner\footnotemark[1]
}

\begin{document}
\maketitle

\begin{abstract}
We study the following general stabbing problem from
a parameterized complexity point of view: Given a set $\mathcal S$ of $n$ translates of an
object in $\Rd$, find a set of $k$ lines with the property that every
object in $\mathcal S$ is ''stabbed'' (intersected) by at least one line.

We show that when $S$ consists of axis-parallel unit squares in $\Rtwo$ the (decision) problem of stabbing $S$ with
axis-parallel lines is W[1]-hard with respect to $k$ (and thus, not fixed-parameter tractable
unless FPT=W[1]) while it becomes fixed-parameter tractable when the squares are disjoint.  
We also show that the problem of stabbing a set of disjoint unit squares in $\Rtwo$ with lines of arbitrary directions is 
W[1]--hard with respect to $k$.  
Several generalizations to other types of objects and lines with arbitrary directions are also presented. 
Finally, we show that deciding whether a set of unit balls in $\Rd$ can be stabbed by one line is W[1]--hard with respect to 
the dimension $d$. 

\end{abstract}



\section{Introduction}
We study several instances of the following general stabbing problem from
a parameterized complexity point of view: Given a set $\mathcal S$ of $n$ translates of an
object in $\Rd$, find a set of $k$ lines with the property that every
object in $\mathcal S$ is ''stabbed'' (intersected) by at least one line. 
Examples include the problem of stabbing a set of axis-parallel squares or circles in
the plane with $k$ lines (possibly axis-parallel), stabbing cubes in space
with $k$ planes, and stabbing unit balls in $\R^{d}$ with one line (the
decision version of the problem of computing the smallest enclosing
cylinder).

All these problems are known to be NP-hard and for most of them only 
polynomial time constant-factor approximation algorithms are known up to date.
We study several such problems from a \emph{parameterized complexity} point of
view: Our goal is to determine if algorithms that run in $O(f(k,d)\cdot
n^{c})$ time on inputs of size $n$ (where $f$ is a computable function
depending only on $k,d$, and $c$ is a constant independent of $k,d,n$) do
exist.

\paragraph{Parameterized Complexity}
We first review some basic definitions of parameterized complexity theory;
see~\cite{DF99,FG06} for an introduction.
A problem with input size $n$ and a positive integer parameter~$k$ is
\emph{fixed-parameter tractable} if it can be solved by an algorithm that
runs in $O(f(k)\cdot n^{c})$ time, where $f$ is a computable function
depending only on $k$, and $c$ is a constant independent of $k$; such an
algorithm is (informally) said to run in fpt-time.  The class of all
fixed-parameter tractable problems is denoted by \fpt. An infinite
hierarchy of classes, the W-hierarchy, has been introduced for establishing
fixed-parameter intractability.  Its first level, \wone, can be thought of
as the parameterized analog of NP: a parameterized problem that is hard for
\wone\ is not in \fpt\ unless \fpt=\wone, which is considered highly
unlikely under standard complexity theoretic assumptions.  Hardness is
sought via an \emph{fpt-reduction}, i.e., a fpt-time many-one reduction
from a problem $\Pi$, parameterized with $k$, to a problem $\Pi'$,
parameterized with $k'$, such that $k'\leq g(k)$ for some computable
function $g$.

\paragraph{Results}

Our results are given by the following theorems listed in the order in which they are proved in the relevant sections.

\begin{theorem}\label{thm:wonehardunitsquaresap} 
Stabbing a set of \ap unit squares in the plane with $k$ \ap lines is \wone--hard with respect to $k$.
\end{theorem}

We prove this by an fpt-reduction from the $k$-\clique\ problem in directed
graphs, which is known to be W[1]-complete~\cite{FG06}. 
This main construction is modified to work for the case when the lines can have arbitrary directions, and 
by replacing the squares with rectangles in a proper way, we get the following theorem:

\begin{theorem}\label{thm:woneharddisjointrectanglesap}  
Stabbing a set of disjoint rectangles in the plane with $k$ lines is \wone--hard with respect to $k$, 
for both cases where the lines are \ap or have arbitrary directions.
\end{theorem}

By simply applying a linear transformation, this leads to the following theorem, 
which complements the results of Langerman and Morin~\cite{DBLP:journals/dcg/LangermanM05}, 
who showed that the same problem for points is fixed parameter tractable.

\begin{theorem}\label{thm:woneharddisjointunitsquares}
Stabbing a set of disjoint unit squares in the plane with $k$ lines of arbitrary directions is \wone--hard with respect to $k$.
\end{theorem}

These theorems are generalized to a large class of objects (for example, squares, circles,
triangles).

\begin{theorem}\label{thm:woneharddisjoint} Let $O$ be a connected object in the plane. (i) If the stabbing lines are to be parallel to two different directions $u,v$ that are part of the input, the problem of stabbing a set of disjoint translates of $O$ with $k$ lines is \wone--hard with respect to $k$, unless $O$ is contained in a line parallel to $u$ or $v$. (ii) The problem of stabbing a set of disjoint translates of $O$ with $k$ lines in arbitrary directions is \wone--hard \wrt $k$.
\end{theorem}

In contrast to the above, some special cases of the problem become fixed parameter tractable.
Let $\mathcal D$ be set of directions.
A line with a direction from $\mathcal D$ is called a $\mathcal D$-line.
A set of objects with the property that the maximum number of
objects that can be simultaneously intersected by two $\mathcal D$-lines
with different directions is bounded by $c\in \mathbb{N}$ is called
\emph{$c$--shallow for $\mathcal D$}. E.g., if we consider the case of \ap disjoint unit squares and \ap lines, the resulting sets are 1--shallow.

\begin{theorem}\label{thm:fpt} 
  (i) Stabbing a set of $n$ \ap disjoint unit squares with $k$ \ap lines is fixed parameter tractable. (ii) The problem of stabbing a set of $n$ translates of a planar connected
  object $O$ that is $\order(1)$-shallow for a set $\mathcal D$ of $\order(1)$ many
  directions with $k$ $\mathcal D$-lines can be decided in $\order(n \log n)$
  time for every fixed $k$.
\end{theorem}
Our algorithm is based on simple data reduction and branching rules that lead to a problem kernel.

Again on the negative side, we show the following:
\begin{theorem}\label{thm:balls}
Stabbing $n$ unit balls in $\R^{d}$ with one line is 
W[1]--hard with respect to $d$.
\end{theorem}
Note that since the balls are unit, the above problem is the decision version of the minimum enclosing cylinder problem.
We prove this result by an fpt-reduction from the $k$-independent set problem in
general graphs, which is known to be W[1]-complete~\cite{FG06}.
In the reduction, the dimension $d$ is linear in the size $k$ of the independent set, hence
an $n^{o(d)}$-time algorithm for this problem implies an $n^{o(d)}$-time algorithm for the $k$-independent set   
problem, which in turn implies that $n$-variable $3$SAT can be solved in $2^{o(n)}$-time. 
The Exponential Time Hypothesis (ETH)~\cite{DBLP:journals/jcss/ImpagliazzoP01} conjectures that no such algorithm exists.

Table~\ref{tab:OurResults} summarizes our results in $\R^{2}$.
The numbers refer to the theorems that prove the corresponding case. 
If no reference is given, the result is trivially implied by the result on its left side.
\begin{table}[h]
	\centering
		\begin{tabular}{|l|c|c|c|c|}\hline
	& \ap  & two dir. fixed & two dir. input & arbitrary  \\	
			\hline unit squares & \wone--h (\ref{thm:wonehardunitsquaresap})& \wone--h & \wone--h & \wone--h (\ref{thm:woneharddisjointunitsquares}) \\
			\hline disj. unit sq. & \fpt (\ref{thm:fpt} (i)) & \fpt (\ref{thm:fpt} (ii)) & \wone--h (\ref{thm:woneharddisjoint} (i)) & \wone--h (\ref{thm:woneharddisjointunitsquares})\\
			\hline disj. rect. fixed & \fpt (\ref{thm:fpt} (ii)) & \fpt (\ref{thm:fpt} (ii)) & \wone--h (\ref{thm:woneharddisjoint} (i))& \wone--h (\ref{thm:woneharddisjoint} (ii))\\
			\hline disj. rect. input & \wone--h  (\ref{thm:woneharddisjointrectanglesap})& \wone--h & \wone--h & \wone--h (\ref{thm:woneharddisjoint} (ii)) \\
			\hline
		\end{tabular}
	\caption{Our results. Term `fixed' refers to the case where the objects or line directions are not part of the input.}
	\label{tab:OurResults}
\end{table}

\paragraph{Related Results}

The parameterized complexity of geometric problems has not been studied
extensively in the past. Some recent examples include work about Klee's
measure problem~\cite{1377693}, clustering~\cite{CGKR08, marx-esa2005}, and
shape-matching~\cite{CGK06}. The survey by Giannopoulos et al.~\cite{DBLP:journals/cj/GiannopoulosKW08} 
provides an extensive overview of
the known results in the area.

The problem of stabbing (or hitting) unit balls in $\Rd$ with one line was
show to be NP-hard when $d$ is part of the input by Megiddo~\cite{Megiddo90onthe}; unless P=NP, the paper
also rules out the existence of a polynomial time approximation scheme for
this problem. This problem is equivalent to the minimum enclosing cylinder problem for points, see Varadarajan et al.~\cite{VVYZ07}. 
Exact and approximation algorithms for the latter problem can be found, for example, 
in B\u{a}doiu et al.~\cite{BHI02}.

Langerman and Morin~\cite{DBLP:journals/dcg/LangermanM05} showed that an abstract $NP$-hard
covering problem that models a number of concrete geometric (as well as
purely combinatorial) covering problems is in \fpt. One example is the
problem of deciding if a set of $n$ points in the plane can be covered
(stabbed) by $k$ lines.

Hassin and Megiddo~\cite{HM91} showed that stabbing line segments with \ap lines is NP--hard even when the 
segments are unit and horizontal. They also developed various constant factor approximation algorithms
for stabbing sets of translates of a given object in the plane and in higher dimensions with \ap lines.

Recently, and independently of our work, 
Dom et al.~\cite{DFR09} considered the parameterized complexity of a stabbing problem similar to ours:
Given a set of \ap lines in the plane and a set of \ap rectangles, find a set of $k$ of those lines that
stab the rectangles. They showed that this problem is \wone--hard (the reduction produces rectangles of different sizes). 
They also claim that this is true even for \ap squares and 
that the problem is in W[1], however no proofs are given for these two results. Observe that 
their hardness result does not imply Theorem~\ref{thm:wonehardunitsquaresap} since, while stabbing \ap unit squares with \ap lines 
(the lines are not given) obviously reduces to the problem they consider, the converse is not at all obvious.
They also showed that for disjoint rectangles the problem is fixed-parameter tractable under the condition that each rectangle is 
stabbed by the same number of vertical and horizontal lines in the given set.
  



\section{Stabbing with $k$ lines}

\subsection{Hardness Results}\label{section:Hardness}
In this section we present the hardness results. The proofs are by a reduction from the $k$--\clique\ problem for directed graphs, which is shown to be \wone--complete in \cite{DF99}. First, in Section \ref{sssec:wonehardUnitSquaresAP}, we show that the problem of stabbing \ap unit squares with \ap lines is \wone--hard. This construction is then modified to work for the case when the lines can have arbitrary directions. From this, minor modifications are made to prove that for this case, the problem is even hard when the squares are disjoint. Finally, we show that the proofs also work for a large class of other objects. In this section, the objects are assumed to be open, but it is easy to modify the proofs to work for closed objects, too.

\subsubsection{Stabbing \ap unit squares with \ap lines in the
  plane}\label{sssec:wonehardUnitSquaresAP}

From a given graph $G$ we will construct a set $\calS(G, k)$ of \ap unit squares in $\Rtwo$ that can be stabbed by $k' := 6k$ lines if and only if the graph has a $k$--clique. The set will be of size $O(n^2k^2)$ and thus polynomial in both $n$ and $k$.

\paragraph{General Idea}
Let $[n] := \{1, \dots, n \}$ and $G = ([n], E)$ be a simple directed graph with no loops. For clarity of presentation, we first create instances $\calS'(G, k)$ that consist of squares of two different sizes, namely some with side length $n-1$ and some with side length $n$. A minor modification will then make them all have the same size. 

As all the squares placed in $\calS'(G, k)$ have integer coordinates and are open, we can simplify our arguments using the following two observations:

\textbf{Observation 1:}
\emph{All the lines of the form $y = i$ or $x = i$ for $i \in \N$ can be neglected, as they can be replaced by any line of the form $y = i \pm \eps$ or $x = i \pm \eps$, $0 < \eps < 1$, respectively, without intersecting fewer squares.
}\\
and

\textbf{Observation 2:}
\emph{Two lines $y = c$, $y' = c'$ with $i < c, c' < i+1$, $i \in \N$, intersect the same squares, and analogously for vertical lines.}

In the final construction there will be $k$ horizontal and $k$ vertical \textit{double strips}, $S_h^1,\dots, S_h^k$ and $S_v^1,\dots, S_v^k$, respectively, to choose lines from. Out of each of those strips, two ``consistent'' lines will have to be chosen in order to get a solution of the specified size. Around every intersection of two such orthogonal strips, we will place a \textit{gadget}, consisting of a set of squares, within a region of suitable size, as indicated in Figure \ref{fig:StripsOverview}.
\begin{figure}[ht]
	\centering
		\includegraphics[scale=.5]{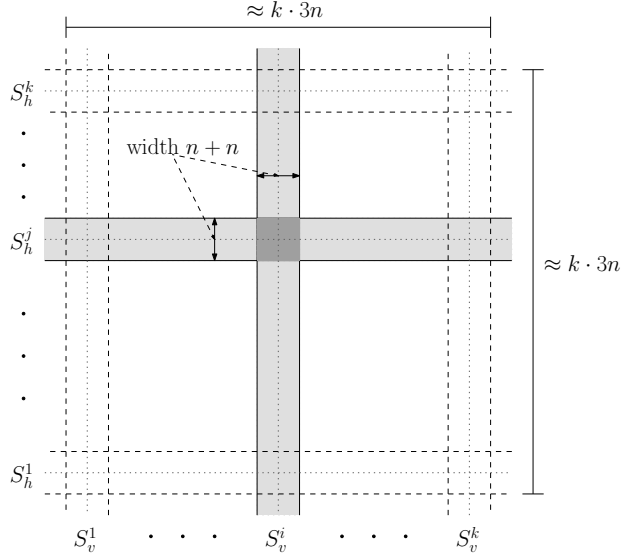}
	\caption{Strips overview}
	\label{fig:StripsOverview}
\end{figure}
We will ensure that any selection of such $4k$ lines has the following properties 
\begin{itemize}
	\item[\textbf{$P_0$:}] Each two lines inside the same double strip will correspond to the same vertex.
	\item[\textbf{$P_1$:}] Two orthogonal line pairs in the strips $S_h^i, S_v^j$, $i \neq j$, will stab all the squares inside the region $S_h^i \cap S_v^j$ if and only if they represent vertices that are connected in $G$.
	\item[\textbf{$P_2$:}] Two orthogonal line pairs in the strips $S_h^i, S_v^i$ will stab all the squares inside the region $S_h^i \cap S_v^i$ if and only if they correspond to the same vertex.
\end{itemize}

Any selection of $2k$ such line pairs will then correspond to a set $C$ of $k$ vertices and will stab all the squares if and only if the vertices in $C$ form a $k$--clique in $G$. 

Besides these $4k$ lines, we will need $2k$ more lines to guarantee the consistency of such a selection ($P_0$). To ensure the properties, several gadgets are constructed, which we will describe in detail now.

\paragraph{The Gadgets} In the following, let $\square_l(x,y)$ denote the \ap square with side length $l$ and lower left corner $(x, y)$. A \textit{gadget} $T$ will consist of a collection of \ap squares. Let $T(x, y)$ denote the copy of $T$ whose squares are placed relative to $(x, y)$. We say that a square is at position $(x', y')$ in gadget $T(x, y)$, if the lower left corner of the square has absolute coordinates $(x + x', y + y')$. Unless stated otherwise, the coordinates of \ap lines are also given relative to the gadget's offset, i.e., if we refer to lines $h\co y = c$ and $v\co x = c$ passing through the gadget $T(x', y')$, we speak about the lines $h\co y = y' + c$ and $v\co x = x' + c$, respectively.

\paragraph{The $F$--Gadget (Forcing)}
The $F$--gadget will be used to ensure that in any solution of size $6k$, a line through a specified strip (of width $n$) must be chosen. We define them as
\[ F_h := \{\square_n(-in, 0) \mid 1 \leq i \leq 6k + 1 \} \]
and 
\[ F_v := \{\square_n(0, -in) \mid 1 \leq i \leq 6k + 1 \}. \]
The fact that they really \textit{force} lines in the specified region follows from the very simple
\begin{proposition}\label{Lemma:Forcing} In order to stab a gadget $F_h(x, y)$ by $6k$ lines, at least one line of the form $y = c$ (relative to the gadget) for some $0 < c < n$ must be chosen.
\end{proposition}
For reasons of symmetry, an analogous proposition holds for the vertical case as well.
We now define the correspondence of lines chosen to vertices in $G$:
\begin{definition} A line $l\co y = c$ through a horizontal $F$--gadget is said to \textit{represent} vertex $\rep(l) := \rup{c} \in V$, and analogously for the vertical case.
\end{definition}
As the $F$--gadgets have a width of $n$, for each vertex in $G$ there exists a line that represents this vertex. Because of Observation 2, two lines that represent the same vertex in a gadget $F$ will intersect the same squares. The (open) strips of width 1 where all the lines represent the same vertex are called \textit{v(ertex)--strips}. Each double strip (of width $2n$) will consist of $2$ vertex strips. See Figure \ref{fig:FGadget}.
\begin{figure}
	\centering
		\includegraphics[width=0.7\textwidth]{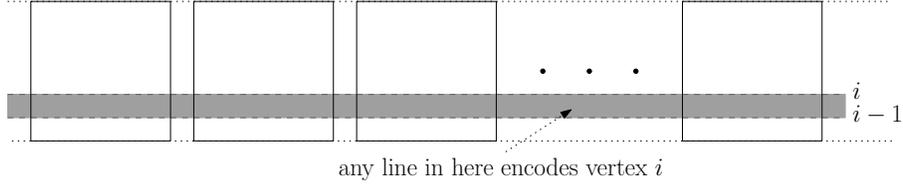}
	\caption{An $F$--Gadget}
	\label{fig:FGadget}
\end{figure}

\paragraph{The $A$--Gadget (Adjacency)}
This gadget represents the adjacency relation of the graph $G$. All the squares will be placed inside a region of size $2n \times 2n$. For each pair of vertices $(i, j)$ such that $(i, j) \notin E$, including the missing loops $(i = j)$, it will contain a square that \textit{forbids} the line pairs corresponding to these vertices to be chosen at the same time, namely $\square_{n-1}(i, j)$.\\
So
\[ A := \{ \square_{n-1}(i, j) \mid (i, j) \notin E \} \]
In the final construction, there will be four $F$--gadgets forcing one line through each of the strips
\begin{itemize}
	\item $S_h^- := \R \times (0, n)$ 
	\item $S_h^+ := \R \times (n, 2n)$
	\item $S_v^- := (0, n) \times \R$ 
	\item $S_v^+ := (n, 2n) \times \R$
\end{itemize}
relative to the gadget's coordinates. $S_h^-$ and $S_h^+$ define a horizontal and $S_v^-$ and $S_v^+$ a vertical double strip.

If a line $l$ lies inside $S_h^-$ or $S_v^-$, it is called \textit{negative}, otherwise it is called \textit{positive}. Two parallel lines are called \textit{antipodal} if one is negative and the other is positive. In the final construction, it will be ensured that if a negative line is chosen that represents vertex $i$, then, in the same double strip, a parallel positive line must be chosen that also represents $i$. Such a line pair is then said to represent vertex $i$.

The main property of the $A$--gadget is stated by the next lemma.
\begin{lemma}\label{Lemma:AdjacencyLemma} Two antipodal vertical lines through $A$ that both represent $i$ and two antipodal horizontal lines through $A$ that both represent $j$ intersect all the squares inside $A$ if and only if $(i, j) \in E$.
\end{lemma}
\begin{proof}
If a square $\square_{n-1}(i', j')$ is not intersected by these lines, we must have $i = i'$ and $j = j'$ and thus $(i, j) = (i', j') \notin E$. If, conversely, $(i, j) \notin E$, then the square $\square_{n-1}(i, j)$ is in $A$ but is not intersected by any of these four lines. 
\end{proof}
With this it will be possible to ensure property $P_1$. Observe that, as the graph contains no loops, also $i \neq j$ is ensured.

An example is shown in Figure \ref{fig:AGadgetExample}. There, the directed edges in both directions are drawn as a single undirected edge. The four squares added for the missing loops are not shown.  
\begin{figure}
	\centering
		\includegraphics[width=0.7\textwidth]{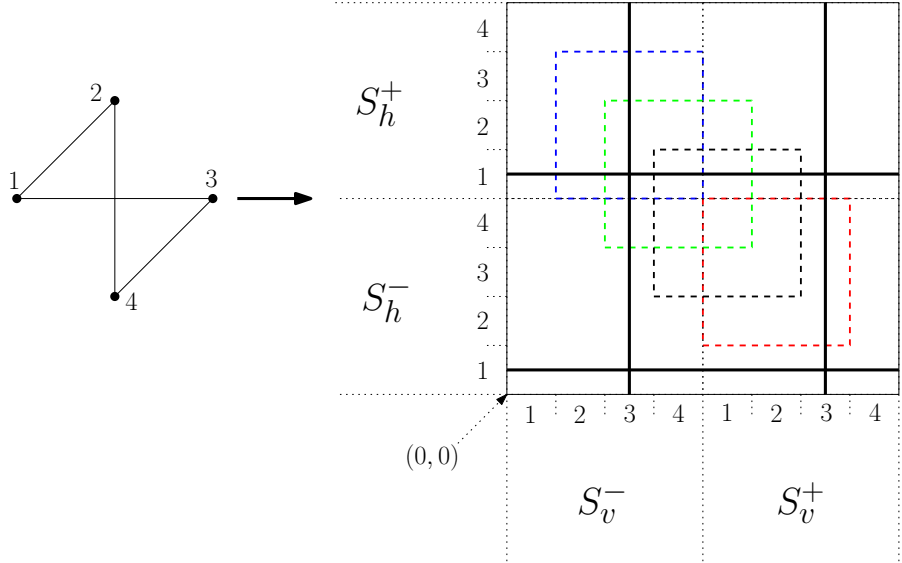}
	\caption{An $A$--Gadget with two antipodal pairs, representing $3$ and $1$, indicated}
	\label{fig:AGadgetExample}
\end{figure}
\paragraph{The $D$--Gadget (Diagonal)}
This gadget is a special $A$--gadget for the graph with the adjacency defined by the identity matrix $I$. It thus consists of the squares
\[ D := \{\square_{n-1}(i, j) \mid 1 \leq i \neq j \leq n \} \]
and will be used to ensure property $P_2$. The regions forced through such a gadget will be the same as for the $A$--gadgets. Thus, by applying Lemma \ref{Lemma:AdjacencyLemma}, all the squares inside a $D$--gadget are stabbed if and only if the vertical and the horizontal antipodal line pair represent the same vertex.

\paragraph{The $C$--gadget (Consistency)}
This type of gadget will guarantee a certain distance between two antipodal lines of the same direction inside the same double strip. 

It ensures that if a size $6k$ solution contains a negative line $l^-$ that represents vertex $i$, then, in the same double strip, it also contains a positive parallel line $l^+$ that represents the same vertex. Thereby it will be possible to identify such a line pair with the vertex $i$, which will ensure property $P_0$.

We continue to describe the $C$--gadgets for the horizontal case. A $C_h$--gadget consists of the union of the two sets
\[ \{ R_i^- := \square_{n-1}(i, i - n + 1) \mid 1 \leq i \leq n-1 \} \]
and
\[ \{ R_i^+ := \square_{n-1}(i - n, n + i - 1) \mid 2 \leq i \leq n \}. \]
In the final construction there will be three $F$--gadgets that ensure the existence of a line in each of the strips
\begin{itemize}
	\item $S_h^- = \R \times (0, n)$ 
	\item $S_h^+ = \R \times (n, 2n)$
	\item $S_{C_h} = (0, n) \times \R$
\end{itemize}
relative to the placement of the gadget. So, in any solution of size $6k$, through each $C$--gadget there will be three lines. Why two of them are given the same name as the strips for the $A$--gadgets will become clear soon.

As for an $A$--gadget, there are again $2n$ combinatorially different horizontal strips to chose lines from. The following lemma states the main property of the $C_h$--gadgets:
\begin{lemma}\label{Lemma:ConsistencyLemma} Let $h^-, h^+$ be two antipodal horizontal lines in $S_h^-, S_h^+$, respectively, then there exists a vertical line that together with $h^-, h^+$ intersects all of the squares belonging to the $C_h$--gadget if and only if $\rep(h^+) \geq \rep(h^-)$.
\end{lemma}
\begin{proof}
First suppose $1 \leq \rep(h^+) < \rep(h^-) \leq n$. Then the two squares $R_{\rep(h^-) - 1}^-$ and $R_{\rep(h^+) + 1}^+$ are defined and are both not stabbed by these two lines. But as 
\[ \underbrace{(\rep(h^-) - 1)}_{\text{left end of } R_{\rep(h^-) - 1}^-} \geq \rep(h^+)
 = \underbrace{(\rep(h^+) + 1) - n + (n-1)}_{\text{right end of } R_{\rep(h^+) + 1}^+} \]
they cannot be stabbed by a single vertical line (recall that the squares are open). See Figure \ref{fig:ConsistencyGadgetFailed}.

If, conversely, $n > \rep(h^+) \geq \rep(h^-) > 1$ (if either $\rep(h^+) = n$ or $\rep(h^-) = 1$, it is trivial), we have that
\begin{eqnarray*} 
D & :=   & \bigcap_{R_i^- \notin I(h^-)}\pr_x(R_i^-) \cap \bigcap_{R_i^+ \notin I(h^+)}\pr_x(R_i^+) \\
  & =    & \pr_x(R_{\rep(h^-) - 1}) \cap \pr_x(R_{\rep(h^+) + 1}) \\
  & \neq & \emptyset
\end{eqnarray*}
as 
\[ \rep(h^-) - 1 < \rep(h^+) = \rep(h^+) + 1 - n + (n-1), \] 
i.\,e., the left side of every $R^-$--square that is not stabbed is to the left of the right side of every $R^+$--square that is not stabbed. Thus, all the squares left can be stabbed by a single vertical line, namely any line of the form $x = c$ for $c \in D$. See Figure \ref{fig:ConsistencyGadgetSuccess}. 
\end{proof}

\begin{figure}
	\begin{minipage}[t]{0.5\textwidth}
		\centering \includegraphics[scale=.45]{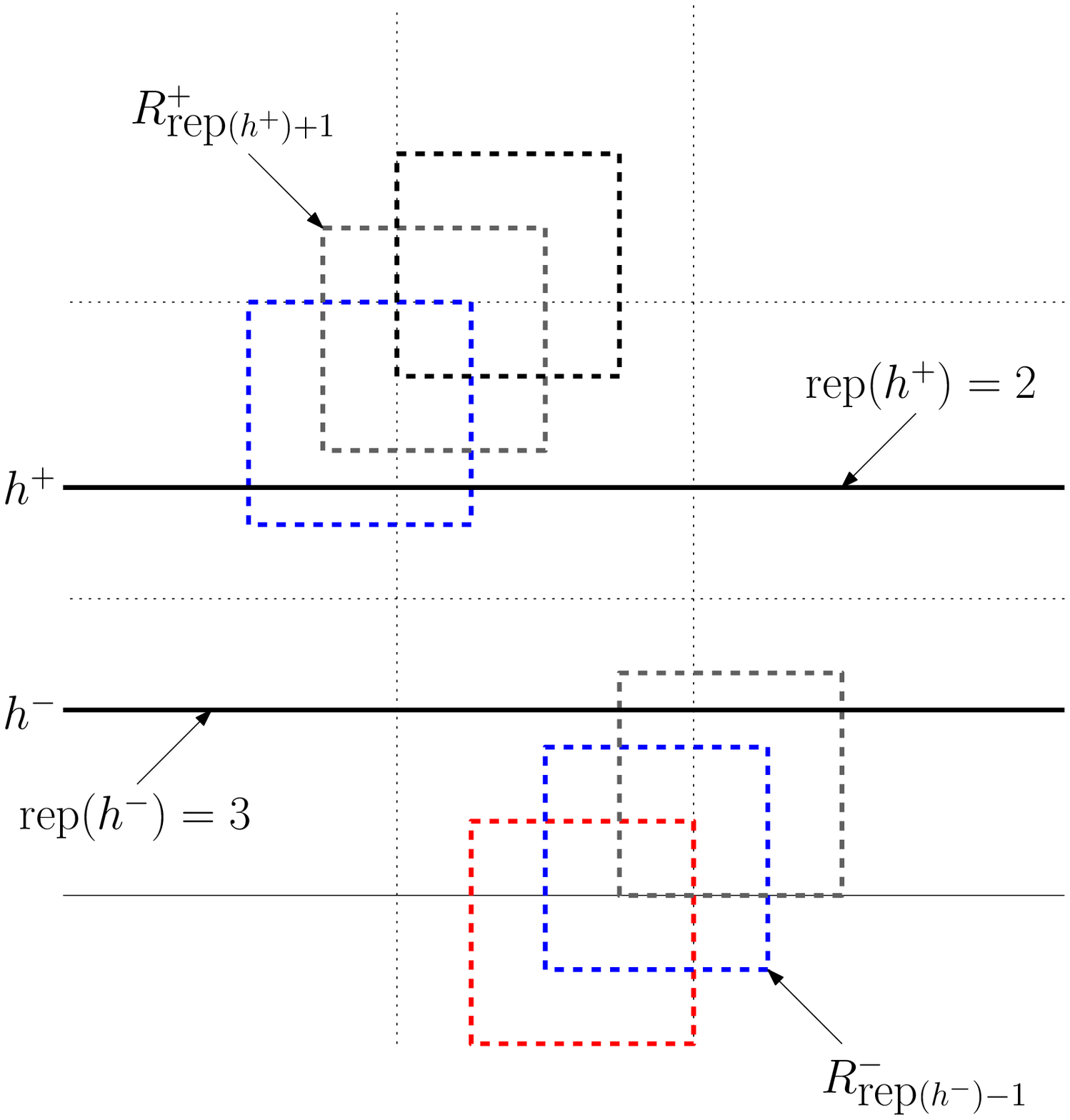}
		\caption{An inconsistent selection}
		\label{fig:ConsistencyGadgetFailed}
	\end{minipage}%
	\begin{minipage}[t]{0.5\textwidth}
		\centering \includegraphics[scale=.45]{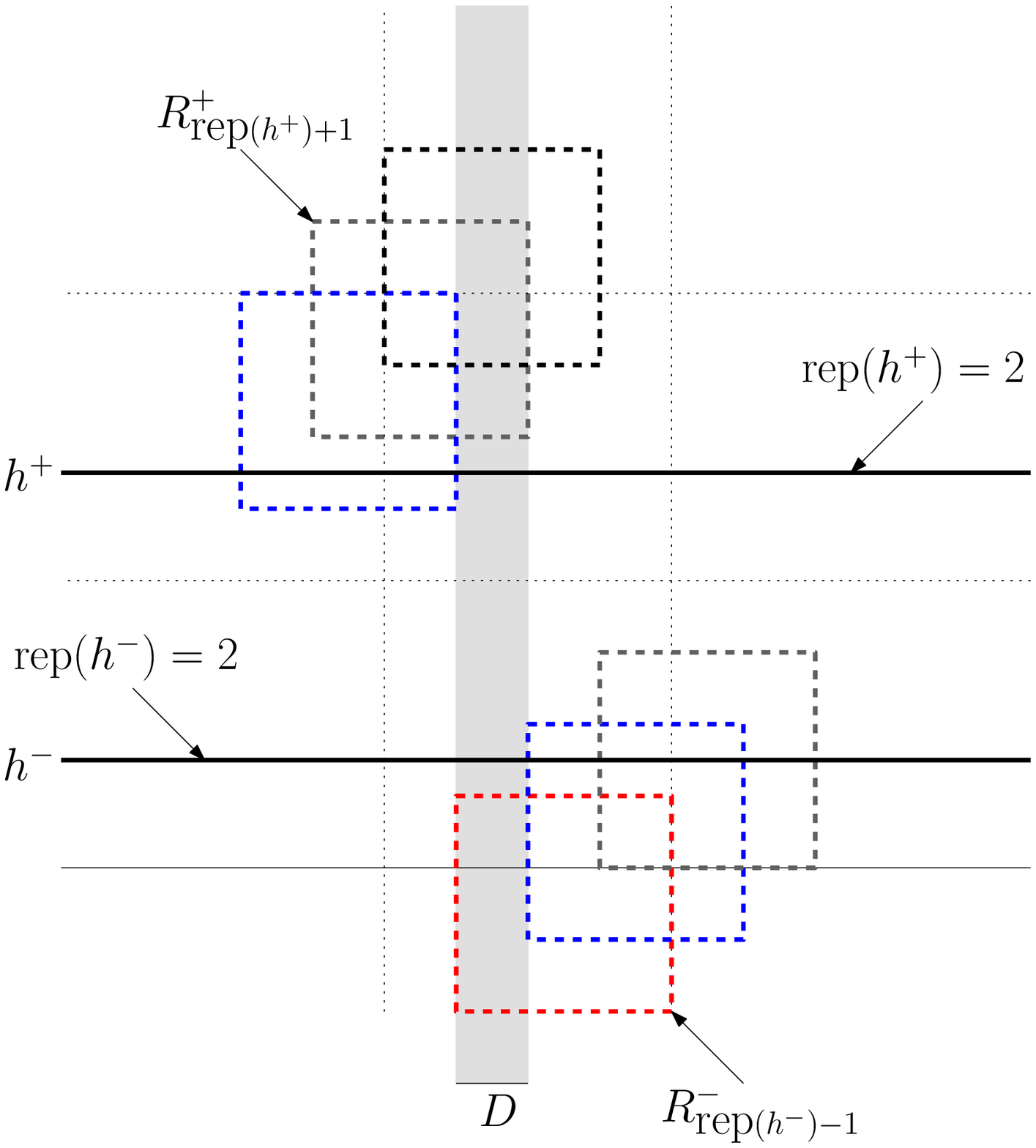}
		\caption{A consistent selection}
		\label{fig:ConsistencyGadgetSuccess}
	\end{minipage}
\end{figure} 
In particular, all squares in a $C$--gadget are intersected if the three lines in $S_h^-, S_h^+, S_{C_h}$ all represent the same vertex.

For the sake of completeness, we give the exact coordinates of the $C_v$--gadgets:
\[ \{ \square_{n-1}(i - n + 1, i) \mid 1 \leq i \leq n-1 \}
 \cup \{ \square_{n-1}(n + i - 1, i - n) \mid 2 \leq i \leq n \}. \]

\paragraph{The Construction}
We now come to describe the exact placement of the gadgets. The main part, expressing the adjacency relation of the graph, will be a $k \times k$ grid of $A$-- and $D$--gadgets:
\[ \mathcal A := \{ A_{i, j} := A(i \cdot 3n, j \cdot 3n) \mid 1 \leq  i \neq j \leq k \}\]
\[ \mathcal D := \{ D_i := D(i \cdot 3n, i \cdot 3n) \mid 1 \leq  i \leq k \}\]
Around this grid, we add the $C$--gadgets to allow only specific solutions:
\[ \mathcal C_h := \{ C_h^i := C_h(-i \cdot 3n, i \cdot 3n) \mid 1 \leq i \leq k\}\]
\[ \mathcal C_v := \{ C_v^i := C_v(i \cdot 3n, -i \cdot 3n) \mid 1 \leq i \leq k\}\]
Here it becomes clear why we chose the coordinates as multiples of $3n$: The $C$--gadgets now cannot influence each other, i.e., no square from one such gadget intersects any strip belonging to another $C$--gadget.
 
Finally, we place the $F$--gadgets to force lines in the desired strips as follows: For the double strips, the lines are forced by
\[ \mathcal S_h^- := \{ (S_h^-)^i := F_h(-3n \cdot (k + 1), i \cdot 3n) \mid 1 \leq i \leq k \}, \]
\[ \mathcal S_h^+ := \{ (S_h^+)^i := F_h(-3n \cdot (k + 1), i \cdot 3n + n) \mid 1 \leq i \leq k \}, \] 
and
\[ \mathcal S_v^- := \{ (S_v^-)^i := F_v(i \cdot 3n, -3n \cdot (k + 1) \mid 1 \leq i \leq k \} \]
\[ \mathcal S_v^+ := \{ (S_v^+)^i := F_v(i \cdot 3n + n, -3n \cdot (k + 1)) \mid 1 \leq i \leq k \}. \] 
The additional lines for the $C$--gadgets are forced by
\[ \mathcal S_{C_h} := \{ S_{C_h}^i := F_v(-i \cdot 3n, -3n \cdot (k + 1)) \mid 1 \leq i \leq k \} \]
and
\[ \mathcal S_{C_v} := \{ S_{C_v}^i := F_h(-3n \cdot (k + 1), -i \cdot 3n) \mid 1 \leq i \leq k \}. \]
The entire construction is shown in Figure \ref{fig:FinalAPConstruction}, where the three regions $(S_h^-)^1, (S_h^+)^1, S_{C_h}^1$ belonging to $C_h^1$ are indicated.
\begin{figure}
	\centering
		\includegraphics[scale=.6]{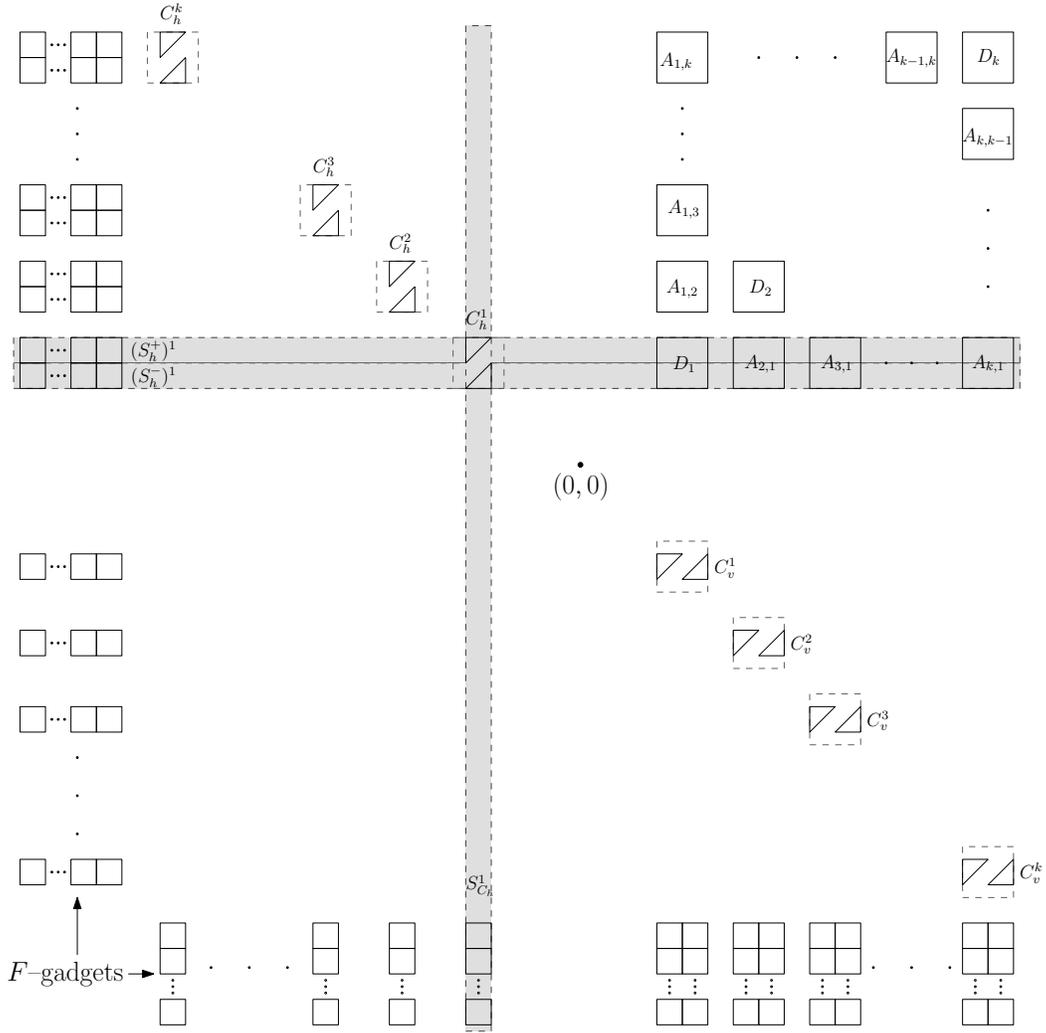}
	\caption{The Final Construction}
	\label{fig:FinalAPConstruction}
\end{figure}
The set
\[ \calS'(g, k) = \mathcal A \cup \mathcal D \cup \mathcal C_h \cup \mathcal C_v \cup \mathcal S_h^- \cup \mathcal S_h^+ \cup \mathcal S_v^- \cup \mathcal S_v^+ \cup \mathcal S_{C_h} \cup \mathcal S_{C_v} \] 
is of size $O(n^2k^2)$ and takes time polynomial in both $n$ and $k$ to create.
\\
It has the following property:
\begin{lemma}
  \label{lem:mainLemma}
  $\mathcal{S}'(G,k)$ can be stabbed by $6k$ axis-parallel lines if and only if
  $G$ has a $k$--clique.
\end{lemma}
\begin{proof}
	Observe that the horizontal as well as the vertical $F$--gadgets are pairwise disjoint, so by Lemma \ref{Lemma:Forcing}, at least one line in the corresponding direction is needed for each of them.
  Thus, in any solution there have to be at least $6k$ lines.  

  Let $G$ have a $k$--clique $C=\{i_1, \dots, i_k\}$. First, we choose $4k$ lines as follows: For $1 \leq j \leq k$, we	choose the line pairs $h_j = (h_j^-, h_j^+)$ (horizontal) and $v_j = (v_j^-, v_j^+)$ (vertical) in the strips $(S_h^-)^j, (S_h^+)^j$ and $(S_v^-)^j, (S_v^+)^j$, respectively, such that they are antipodal and correspond to the vertex $i_j$.
  
  Then we have, for parallel lines, that
  $\rep(l_j^-) = \rep(l_j^+)$ ($l \in \{h, v\}$) and thus we can apply Lemma
  \ref{Lemma:ConsistencyLemma}, i.\,e., the squares left in the $2k$ $C$--gadgets can be
  intersected by $2k$ additional lines. By Lemma \ref{Lemma:AdjacencyLemma}, all the squares
  inside $A_{j, m}$ are intersected, as $(i_j, i_m) \in E$ for all $j
  \neq m$. Further, as $h_j$ and $v_j$ represent
  the same vertices, all $D$--gadgets are also stabbed. 
  Thus, $6k$ lines suffice.

  Now assume that the set can be stabbed by $6k$ axis-parallel lines. Because of the
  $F$--gadgets, through each $A$-- and $D$--gadget there must be exactly
  two antipodal horizontal and two antipodal vertical lines. Also, through each $C$--gadget there are exactly three lines, two of which are parallel.\\ Further, by Lemma \ref{Lemma:ConsistencyLemma} we have for each such antipodal pair $l_j^-,
  l_j^+$ of lines in the same double strip that $\rep(l_j^+) \geq \rep(l_j^-)$, 
  for otherwise the corresponding $C$--gadget would not be stabbed. We can
  assume that $\rep(l_j^-) = \rep(l_j^+)$ for all $1 \leq j \leq k$: enlarging the gap between the two antipodal parallel 
  lines can only reduce the set of squares that are intersected in the $A$-- and $D$--gadgets, and, by Lemma \ref{Lemma:ConsistencyLemma}, the additional line in the $C$ gadgets can then be chosen to represent $\rep(l_j^-)$, too. Each such pair of lines thus corresponds
  to a node in $G$ ($P_0$). Let $C=\{i_1, \dots, i_k\}$ be the nodes
  represented by the horizontal line pairs and $C'=\{i_1', \dots, i_k'\}$
  the nodes represented by the vertical line pairs. By Lemma \ref{Lemma:AdjacencyLemma}, the gadget $D_j$ ensures that $i_j = i_j'$ for all $1 \leq j \leq k$ ($P_2$), and thus we have $C = C'$. Further, the gadget $A_{j, m}$
  ensures that $(i_j, i_m) \in E$ for all $j \neq m$ ($P_1$), which also
  implies $i_j \neq i_m$ for all $j \neq m$ as the graph contains no loops. But this means that $C$ forms a $k$--clique in $G$.
  
\end{proof}

\paragraph{Adaption to Unit Squares}
To make all the squares have a side length of $n-1$, we simply shrink the squares inside the $F$--gadgets by $1/2$ from each side, i.\,e. we redefine the $F$--gadgets as 
\begin{itemize}
	\item $F_h := \{ \square_{n-1}(-in + 1/2, 1/2) \mid 1 \leq i \leq 6k+1 \}$ and 
	\item $F_v := \{ \square_{n-1}(1/2, -in + 1/2) \mid 1 \leq i \leq 6k+1 \}$.
\end{itemize}
and define $\calS(G, k)$ accordingly. The only lines influenced by this are the ones that represent either $1$ or $n$. Because all the lines that represent $1$ in a gadget $F_h$ intersect the same squares in $\calS'(G, k)$, we can assume that any such line in a solution is of the form $y = 3/4$. The same argument holds for the lines that represent $n$, i.\,e., they can assumed to be of the form $y = n - 3/4$; again, the same holds analogously for vertical lines. Thus, if there is a solution of size $6k$ for $\calS'(G, k)$, then there is also one for $\calS(G, k)$. This completes the proof of Thm. \ref{thm:wonehardunitsquaresap}.

\subsubsection{Arbitrary directions}\label{sssec:arbitrary}

So far our results depended on the lines being parallel to the coordinate axis. In this section, starting with the set $\calS(G, k)$ of \ap unit squares from section \ref{sssec:wonehardUnitSquaresAP}, we show how to modify this construction to yield a set $\calS^*(G, k)$ of \ap unit squares that works for the case where the lines can lie in arbitrary directions. Observe that, while intuitively plausible, it is not a priori clear that this problem is also \wone--hard just because the problem for \ap lines is hard.

The proof that this problem is hard is more technical than above, even though the idea remains the same. The main task will be to modify the set $\calS(G, k)$ in such a way that the lines in any solution must be ``almost'' axis--parallel. This will be done by increasing the number of squares the $F$--gadgets and shrinking the squares a little. With this it will be possible to show that all the almost \ap lines have an equivalent \ap line. 

To make calculations easier, we first modify $\calS(G, k)$ by applying the linear function that scales in $x$-- and $y$--direction by $1/n$. If we now refer to $\calS(G, k)$, we mean the scaled set. All the squares in this set have side length $u := (n - 1)/n = 1 - 1/n$. The vertex--strips for $2, \dots, n-1$ then have a width of $s := 1/n$, and the vertex--strips for $1$ and $n$ have a width of $s/2$. 

\paragraph{Shrinking the Squares}
\emph{To shrink a square by $\eps$} means that we replace a square $\square_l(x, y)$ by $\square_{l - 2\eps}(x + \eps, y + \eps)$, i.e., shrink it from each side by a value of $\eps$. We begin with the definition of $\delta$--robustness which will prove to be very useful in the following argumentation.
\begin{definition} A set $S$ of squares is called \textit{$\delta$--robust}, if
\[ \forall R \subseteq S: \bigcap_{r \in R}\pr_d(r) \neq \emptyset \follows \diam\left(\bigcap_{r \in R}\pr_d(r)\right) \geq 2\delta \]
for $d \in \{x, y\}$.
\end{definition}
A set that is $\delta$--robust can be altered a little without ``destroying'' any solutions. The following lemma will be used in its full strength in the next section. During this section, we will only consider modifications that shrink the squares. 
\begin{lemma}\label{Lemma:DeltaRobust} Let $S$ be a $\delta$--robust set of \ap unit squares that can be stabbed by $k$ \ap lines. If we translate each square by a value at most $\tau$ and shrink each square by a value $\sigma$ such that $\tau + \sigma < \delta$, then the resulting set still can be stabbed by $k$ \ap lines. Further, if all the squares are shrinked by $\sigma < \delta$ (and not translated), then the resulting set is $(\delta - \sigma)$--robust.
\end{lemma}
\begin{proof} For a set $R \subseteq S$, let $R^*$ denote the modified set. Obviously, for any set of squares $R$ we have $\bigcap_{r \in R}\pr_d(r) \neq \emptyset$, $d \in \{x, y\}$, if and only if there exists an \ap line that stabs all squares from $R$. We show that the modified set $R^*$ still lies on a common \ap line. Let $l$ be horizontal and 
\[ y_{min} := \inf \bigcap \{ \pr_y(r) \mid r \in R \}, y_{max} :=  \sup \bigcap \{ \pr_y(r) \mid r \in R \}. \] 
The line $l'\co y = y_{min} + \frac{1}{2}\left(y_{max} - y_{min}\right)$ intersects all the squares from $R$. Further, $y_{max} - y_{min} \geq 2\delta$, as the set is $\delta$--robust. Thus, after shrinking and translating the squares in $R$ by a value of at most $\tau$ and $\sigma$, respectively, for the corresponding values $y_{min}^*, y_{max}^*$ of the modified set $R^*$ we still have 
\[ y_{max}^* - y_{min}^* \geq y_{max} - (\tau + \sigma) - \left(y_{min} + (\tau + \sigma) \right) = 2\delta - 2(\tau + \sigma) > 0. \]
Thus, $l'$ stabs all the squares from $R^*$.
Again, the same argument works for vertical lines as well.

To prove the second part, observe that 
\begin{eqnarray*} \bigcap_{r^* \in R^*}\pr_d(r^*) \neq \emptyset & \follows & \bigcap_{r \in R}\pr_d(r) \neq \emptyset\\
& \follows & \diam\left(\bigcap_{r \in R}\pr_d(r)\right) \geq 2\delta\\
& \follows & \diam\left(\bigcap_{r^* \in R^*}\pr_d(r^*)\right) \geq 2\delta - 2\sigma
\end{eqnarray*}
for $d \in \{x, y\}.$

\end{proof}
See Figure \ref{fig:ShrinkingSquares}.
\begin{figure}[ht]
	\centering
		\includegraphics[scale=.4]{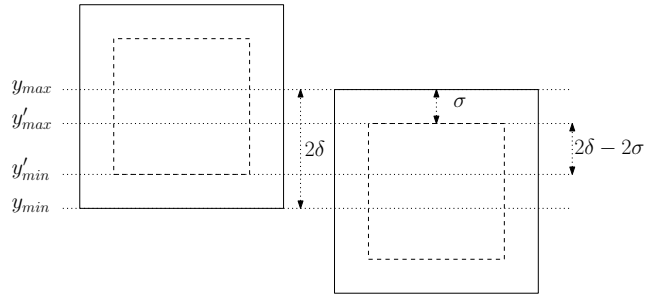}
	\caption{Shrinking the squares (here, $\tau = 0$)}
	\label{fig:ShrinkingSquares}
\end{figure}
(Observe that in general the reverse is not true.)

We will now modify the set $\calS(G, k)$ to yield a set $\calS^*(G, k)$ in two steps as follows:
\begin{enumerate}
	\item The $F$--gadgets are enlarged to now contain $N := n^2$ squares, i.e. we set
	\[ F_h := \{ \square_u(-i + s/2, s/2) \mid 1 \leq i \leq N \} \]
	and
	\[ F_v := \{ \square_u(s/2, -i + s/2) \mid 1 \leq i \leq N \}.\]
        (Recall that we have scaled the set $\calS(G, k)$ by $s = 1/n$ to contain squares of length $u$).
	\item In the resulting set, all the squares are \emph{shrinked} by $\eps := s/6 = 1/(6n)$.
\end{enumerate}

The resulting set then consists of unit squares with side length $u^* := 1 - 1/n - 2\eps$. 
We will make use of the following observation, which is easy to check:

\textbf{Observation 1: }\emph{For any two squares $r, r'$ from $\calS^*(G, k)$ and $d \in \{x, y\}$, we have that \[ \pr_d(r) \cap \pr_d(r') = \emptyset \follows \dist(\pr_d(r), \pr_d(r')) \geq 2\eps. \] }
That means that if two squares cannot be interesected by, e.g., a common vertical line, then there is a horizontal distance of at least $2\eps$ between them.
Lemma \ref{Lemma:DeltaRobust} is used to prove the following property of our set $\calS^*(G, k)$:
\begin{lemma}\label{Lemma:ShrinkingLemma}
The set $\calS^*(G, k)$ can be stabbed by $6k$ \ap lines if and only if $\calS(G, k)$ can be stabbed by $6k$ \ap lines.
\end{lemma}
\begin{proof}
First observe that if we are only considering solutions of size $6k$ with \ap lines, then it does not matter whether the $F$--gadgets consist of $6k+1$ or $N$ squares.\\
``$\Rightarrow$'': The squares from $\calS(G, k)$ all contain a square from $\calS^*(G, k)$, thus any solution to $\calS^*(G,k)$ is a solution to $\calS(G, k)$.\\
``$\Leftarrow$'': By the construction of $\calS(G, k)$, it is $s/4$--robust and $\eps = s/6 < s/4$. Thus, we can apply Lemma \ref{Lemma:DeltaRobust}. 
\end{proof}
By $T^*$ we denote the modified version of gadget $T$, e.\,g., $A^*$ is the $A$--gadget with the squares shrinked as described above. The following proposition is used to show that in any solution of size $6k$ the lines have to be almost parallel to the axis. 
\begin{proposition}\label{Proposition:MaxStab}
A line $l\co ax + by = c$ can intersect at most $\rup{|b/a|} + 1$ squares of a single $F_h^*$ gadget and at most $\rup{|a/b|} + 1$ of a single $F_v^*$--gadget.
\end{proposition}
\begin{proof} For any two points $(x, y)$ and $(x', y')$ where the line stabs a square from an $F_h^*$--gadget, we must have $|y - y'| < u^*$, which means $|x - x'|\cdot |(a/b)| < u^* < 1$ and thus $|x - x'| < |b/a|$. Thus, as the squares inside the $F_h^*$--gadget are all disjoint, at most $\rup{|b/a|} + 1$ of them can be stabbed by such a line. Rotation by 90 degrees shows that for the $F_v^*$--gadgets at most $\rup{|a/b|} + 1$ squares can be stabbed.

\end{proof}

To prove the main property of the lines, we first only consider the set of $6k$ $F^*$--gadget and do not add the $A^*$--, $D^*$--, and $C^*$--gadgets yet. 

As all the squares in $\calS^*(G, k)$ are placed between $x_l = -\left(3(k + 1) + N\right)$, $x_r = 3k + 3$, $y_b = -\left(3(k + 1) + N\right)$, and $y_t = 3k + 3$, it suffices to consider the behaviour of the lines inside the region $(x_l, x_r) \times (y_b, y_t)$.
Then the following holds:
\begin{lemma}\label{Lemma:SlopeLemma} In order to stab the $6k$ $F^*$--gadgets with $6k$ lines in arbitrary directions, each of the lines has to intersect a single $F^*$--gadget entirely.
\end{lemma}
\begin{proof} 
It suffices to show that any line can stab at most $N$ squares and that this is the case only if it stabs a single $F^*$--gadget entirely. As there are $6kN$ squares to stab, the claim follows. Without loss of generality, let $l\co y = mx + c$ for some $|m| \leq 1$; the vertical case is symmetric. We call such a line that stabs $N$ squares an $h^*$--line and show in three steps:
\begin{enumerate}
	\item[a.] An $h^*$--line must have a slope $|m| \leq 4k/N$.
	\item[b.] An $h^*$--line cannot intersect squares from two different $F_h^*$--gadgets.
	\item[c.] An $h^*$--line cannot intersect any squares from an $F_v^*$--gadget.
\end{enumerate}
from which it follows that an $h^*$--line must intersect a single entire $F_h^*$--gadget.
\medskip
\noindent
\\
\textbf{a.} If the slope $|m|$ is larger than $4k/N$, i.\,e., $4k/N < |m| \leq 1$, by Proposition \ref{Proposition:MaxStab} the line can stab at most
\begin{eqnarray*}
3k(\rup{N/(4k)} + 1) + 3k(\rup{4k/N} + 1) & \leq & 3k(N/(4k) + 1 + 1 +2) \\
                                          & = & \frac{3}{4}N + 12k\\
                                          & < & N 
\end{eqnarray*}
squares. So any $h^*$--line must have a slope $|m| \leq 4k/N$.
\medskip
\noindent
\\
\textbf{b.} When such a line intersects a square of one $F_h$--gadget at $x = t_0$, it cannot intersect any square of another $F_h^*$--gadget at $x = t_1$ unless $|(t_1 - t_0)(4k/N)| \geq 2\eps$ (the gap between two $y$--disjoint squares, see Observation 1) and thus $|t_1 - t_0| \geq 6k + 3$ (as $n >> k$). In particular, if such a line intersects the $j$--th square (from the right) of one $F_h^*$--gadget, it cannot intersect the $j'$--th square from another $F_h^*$--gadget for $j - (6k + 1) \leq j' \leq j + (6k+1)$.

Let $C$ denote the number of different $F_h^*$--gadgets intersected. Then the total number of squares stabbed is at most $N - (C-1)(6k + 1) + 6k$, which is less than $N$ for $C > 1$. Thus we have $C = 1$, i.\,e. any $h^*$--line can intersect at most one $F_h^*$--gadget and must stab at least $N - 6k$ of its squares. Thus it must have a slope of at most $|m| \leq 1/(N - 6k - 1) < 2/N$.
\medskip
\noindent
\\
\textbf{c.} In order for a line to stab $N - 6k$ squares of a single $F_h^*$--gadget, it must intersect the $(6k + 1)$--th square (from the right) of this gadget. Thus, at $x = -3(k+1) - 6k - 1$, any $h^*$--line must be above $y = -3k$, which is below the lowest point where it can stab any square from an $F_h^*$--gadget. (Observe that the bounds are even stronger, e.g., any such line must even be above $-3k + s/2 + \eps$, but this is not needed here).  Then the line cannot stab any square from an $F_v^*$--gadget, as 
\[-3k - |x_r - (-3(k+1) - 6k - 1)|\cdot 2/N > -3k,\]
and any square from an $F_v^*$--gadget lies below $-3(k + 1)$. So it must lie entirely inside a single $F_h^*$--gadget in order to be an $h^*$--line. 
Analogous calculations prove the same for the case $|m| > 1$ when the line is almost vertical.
\end{proof}
Figure \ref{fig:CoordinatesOverviewNew} indicates the coordinates used.
\begin{figure}[ht]
	\centering
		\includegraphics[scale=.7]{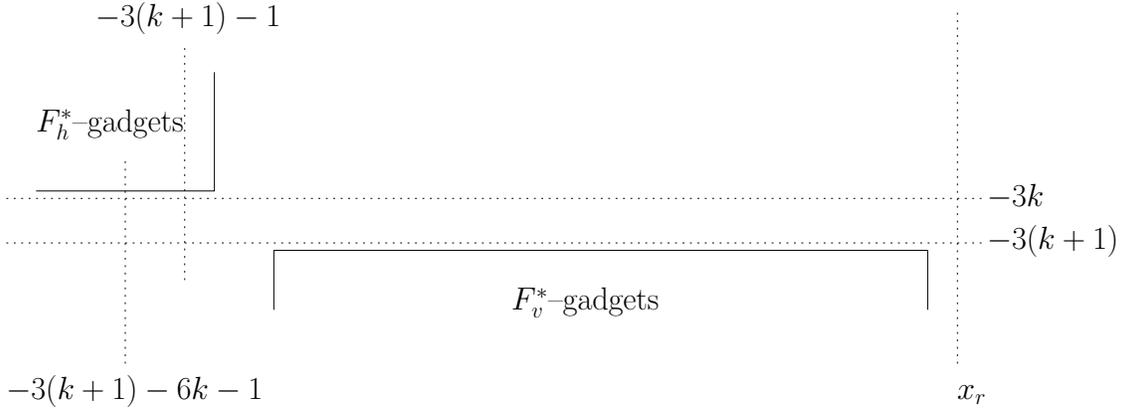}
	\caption{The coordinates}
	\label{fig:CoordinatesOverviewNew}
\end{figure}
Thus, for the $F^*$ gadgets only, we know that in order to stab all the squares with $6k$ lines, one line must intersect exactly one (entire) $F^*$--gadget. In order to do so, by Proposition \ref{Proposition:MaxStab}, it must have a slope of at most $1/(N-1)$ (in the horizontal case) or at least $N-1$ (in the vertical case). The crucial point is  that if we now \emph{add} squares to the existing set, these properties remain.

\paragraph{The Final Construction}
Now we place the remaining squares from $\calS^*(G, k)$. Recall tha, by Lemma \ref{Lemma:ShrinkingLemma}, $\calS^*(G, k)$ can be stabbed by $6k$ \textit{\ap} lines iff $\calS(G, k)$ can be stabbed by $6k$ \ap lines.\\
By shrinking, we have created a small ``fuzzy'' region (see Observation 1) and have thereby achieved that the small change that a line can make after leaving its $F^*$--gadget cannot influence the solution. This is expressed by the next lemma:
\begin{lemma}\label{Lemma:WlogHV} In any solution to $\calS^*(G, k)$ with $6k$ arbitrary lines, without loss of generality the lines can assumed to be axis--parallel, i.\,e., if there is a solution with $6k$ arbitrary lines, then there is also one with $6k$ \ap lines. 
\end{lemma}
\begin{proof} 
Let $l$ be an almost horizontal line with slope $|m| < 1/(N-1)$. As the line has to intersect an entire $F_h^*$--gadget, it suffices to calculate the change it can make between the minimum $x$--position where it can leave an $F_h^*$--gadget, namely $-3(k + 1) - 1 + s/2 + \eps$ (Figure \ref{fig:CoordinatesOverviewNew}), and $x_r$, which is
\[ |x_r - \left(-3(k + 1) - 1 + s/2 + \eps\right) |\cdot |m| < 10k \cdot |1/(N-1)| < 2\eps. \] 
Thus, it cannot intersect any two $y$--disjoint squares, from which it follows that it can be replaced by a horizontal line. Again, similar calculations prove the vertical case.
\end{proof}
That means if there is a solution with arbitrary lines for the set $\calS^*(G, k)$, then there is also one where all the lines are axis--parallel.
Using Lemma \ref{Lemma:ShrinkingLemma}, it follows that $\calS(G, k)$ can be stabbed by $6k$ \ap lines if and only if $\calS^*(G, k)$ can be stabbed by $6k$ arbitrary lines, which proves the following: 
\begin{theorem}\label{thm:wonehardunitsquaresarbitrary} 
Stabbing a set of \ap unit squares in the plane with $k$ lines of arbitrary directions is \wone--hard with respect to $k$.
\end{theorem}

\subsubsection{Sets of disjoint objects}

In this section we show that some of the problems are even hard for sets of disjoint objects. First, we show that stabbing disjoint rectangles with \ap lines is \wone--hard if the rectangles can be chosen arbitrarily. This goes by a small modification of the sets in the previous sections. It is important to notice that for this problem, the rectangle chosen for the reduction, i.e., the ratio of its side lengths, depends on $n$, in contrast to the results in the previous section, where (after scaling the construction) only a single base object was required. 

From this we derive, as a main result, that stabbing \emph{disjoint} \ap unit squares with lines in \emph{arbitrary} directions is also \wone--hard, in contrast to the case where the lines have to be axis--parallel, which is covered in the next chapter.

The proof will consists of three steps which we will sketch here first:
\begin{enumerate}
	\item ``Wobble'' the squares in $\calS^*(G, k)$ a little, such that all the (parallel) diagonals of the squares are disjoint.
	\item Replace each diagonal with a \emph{very thin} rectangle, such that all the resulting rectangles are disjoint.
	\item Transform the set of rectangles to a set of unit squares via a bijective linear transformation.
\end{enumerate}

\subsubsection{Disjoint Rectangles}
Starting with the set $\calS^*(G, k)$ from the previous section, we will construct a set of disjoint rectangles $\mathcal R^*(G, k)$ that can be stabbed by $6k$ lines if and only if the $\calS^*(G, k)$ can. This will prove the hardness for both the cases where the lines chosen have to be \ap as well as for arbitrary lines.

Recall that the squares in $\calS^*(G, k)$ have a side length of $u^* = 1 - 1/n - 2\eps$ for the $\eps$ defined as $s/6$. By Lemma \ref{Lemma:WlogHV}, the set $\calS^*(G, k)$ can be stabbed by $6k$ arbitrary lines if and only if it can be stabbed by $6k$ \ap lines, and by Lemma \ref{Lemma:DeltaRobust}, the set $\calS^*(G, k)$ is $(s/12)$--robust, as $\calS(G, k)$ is $(s/4)$--robust and $s/4 - s/6 = s/12$.

We will modify the set $\calS^*(G, k)$ such that no two (parallel) diagonals intersect any more while maintaining the significant combinatorial properties. Recall that right now for an $A^*$--, $D^*$--, and $C^*$--gadgets, the diagonals of some of the squares may intersect, as indicated in Figure \ref{fig:Wobbling}.
\begin{figure}[ht]
	\centering
		\includegraphics[width=1.00\textwidth]{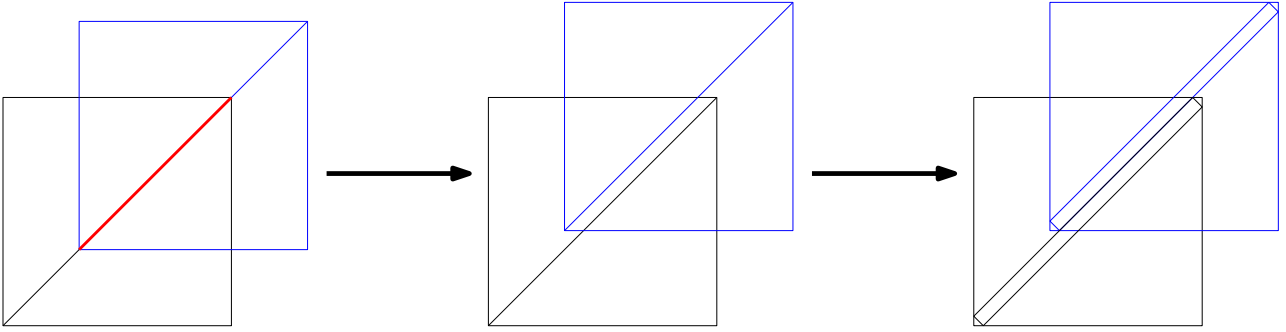}
	\caption{Wobble and replace}
	\label{fig:Wobbling}
\end{figure}

Let $W := n^{-4}$ and $\varphi(i, j) := i \cdot n + j$. The new squares will have a side length of $u_w := u^* - 2Wn^2$. We define the wobble--function $\omega$, which shrinks and translates the squares, as follows:
\[ \omega_{i, j}\left(\square_{u^*}(x, y)\right) = \square_{u_w}(x + Wn^2, y + Wn^2 + \varphi(i,j)\cdot 2W). \]
We now take the set $\calS^*(G, k)$ and wobble the squares inside the $A^*$--, $D^*$, and $C^*$--gadgets. For the $A^*$-- and $D^*$--gadgets, we apply $\omega_{i, j}$ to the square that is added for $(i, j) \notin E$ (which is $\square_{u^*}(i/n + \eps, j/n + \eps)$, relative to the gadget's offset).

Each $C^*$--gadget contains $2n - 2$ squares. For each such gadget, we apply $\omega_{i\text{ div }n, i \text{ mod } n}$ to the $i$--th square.
The other squares, i.\,e., those contained in the $F^*$--gadgets, are simply shrinked (but not shifted) to be all of size $u_w \times u_w$. This yields a set of \ap unit squares $\calW^*(G, k)$.

Now we want replace the diagonals of the squares in $\calW^*(G, k)$ by very thin rectangles, which will be all disjoint. We define the rectangle $\rho_W$ by its endpoints
\[ \rho_W(x, y) := \{ (x + W, y), (x, y + W), (x + u_w - W, y + u_w), (x + u_w, y + u_w - W) \}\]
as shown in Figure \ref{fig:Rectangle}.
\begin{figure}[ht]
	\centering
		\includegraphics[scale=.5]{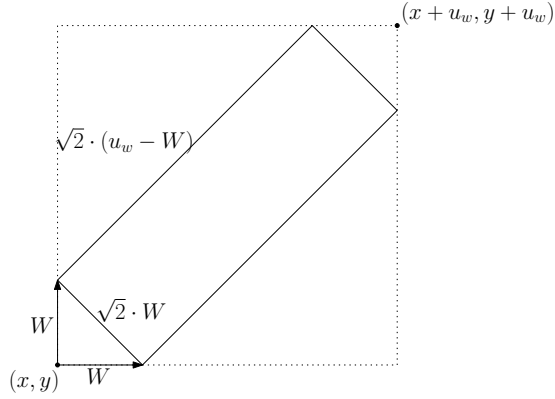}
	\caption{The rectangle $\rho_W(x, y)$}
	\label{fig:Rectangle}
\end{figure}
Instead of each square in $\calW^*(G, k)$ we now place a rectangle $\rho_W$ whose bounding box is this square.

Thereby we have achieved that all the rectangles created (which are all copies of $\rho_W$) are disjoint, as the distance of two diagonals is now at least $\sqrt 2 \cdot W = 2 \cdot \frac{1}{2} \underbrace{\sqrt 2 \cdot W}_{\text{width of } \rho_W}$. Thus, the resulting set $\calR^*(G, k)$ is a set of disjoint translates of $\rho_W$.

Now we can show the main lemma of this section, which completes the proof of Thm. \ref{thm:woneharddisjointrectanglesap}.
\begin{lemma}\label{lemma:disjointRectanglesAP} $\calS^*(G, k)$ can be stabbed by $6k$ lines if and only if $\calR^*(G, k)$ can be stabbed by $6k$ lines.
\end{lemma}
\begin{proof} We prove that the following are equivalent
\begin{enumerate}[(i)]
	\item $\calS^*(G, k)$ can be stabbed by $6k$ arbitrary lines.
	\item $\calS^*(G, k)$ can be stabbed by $6k$ \ap lines.
	\item $\calR^*(G, k)$ can be stabbed by $6k$ \ap lines.
	\item $\calR^*(G, k)$ can be stabbed by $6k$ arbitrary lines.
\end{enumerate}
(i) $\follows$ (ii): By Lemma \ref{Lemma:WlogHV}.\\
(ii) $\follows$ (iii): Obviously, an \ap line intersects a square iff and only if it intersects its inscribed rectangle $\rho_W$. As the set $\calS^*(G, k)$ is $s/12$--robust and
\[ \underbrace{2Wn^2}_{\text{max. shift}} + \underbrace{2Wn^2}_{\text{shrink}} = 4 n^{-2} < 1/\left(12n\right) = s/12, \] 
we can apply Lemma \ref{Lemma:DeltaRobust}.\\
(iii) $\follows$ (iv): trivial\\
(iv) $\follows$ (i): All the wobbled squares are contained in the original squares, as the maximum shift is $Wn^2$ and they are shrinked by $Wn^2$ from each side. Thus, any solution to the set of inscribed rectangles $\calR^*(G, k)$ is also a solution to $\calS^*(G, k)$.
\end{proof}

\subsubsection{Disjoint Unit Squares}
To prove the case of disjoint unit squares now is an easy task. The matrix 
\begin{eqnarray*}
M & = & \frac{1}{\sqrt 2}
\left(
\begin{array}{cc}
1/W & 0\\
0 & 1/\left(u_w - W\right)
\end{array} 
\right)
\cdot\frac{1}{\sqrt 2}
\left(
\begin{array}{cr}
1 & -1\\
1 & 1
\end{array}
\right)\\
& = & \frac{1}{2}
\left(
\begin{array}{cc}
1/W & -1/W\\
1/(u_w - W) & 1/(u_w - W)
\end{array} 
\right)
\end{eqnarray*}
represents a bijective linear transformation and the image of $\rho_W$ under $M$ is an \ap unit square. Thus, the set $\calU^*(G, k) := M\cdot \calR^*(G, k)$ consists of disjoint unit squares and is combinatorially equivalent to $\calR^*(G, k)$. This leads to the proof of Thm. \ref{thm:woneharddisjointunitsquares}. Also, observe that because of Lemma \ref{lemma:disjointRectanglesAP}, $\calU^*(G, k)$ can be stabbed by $6k$ lines in direction either $M\cdot e_1$ or $M \cdot e_2$, where $e_i$ denotes the canonical base vector, if and only if it can be stabbed by $6k$ arbitrary lines. This will be used for the proof of Thm. \ref{thm:woneharddisjoint} in the next section.

\subsubsection{Other objects}

Using the results from the previous sections, we now come prove the \wone--hardness for a wide range of stabbing problems. The objects we will consider are those which, from two directions, ``look like a square''. This can be formalized as follows:
\begin{definition} Let $d, d'$ be two linearly independent vectors. An object $o$ is said to be a \textit{quasi--square} with respect to $d$ and $d'$, if the projection of $o$ on each of the orthogonal complements of $d$ and $d'$ is an open line segment, i.\,e., is homeomorphic to $(0, 1)$.
\end{definition}
For an object $o$, we define the \ap \emph{bounding box $\BB(o)$} as
\[ \BB(o) := \pr_{x}(o) \times \pr_{y}(o). \]
Obviously, if $\pr_x(o)$ and $\pr_y(o)$ are connected, an \ap line intersects the bounding box of an object if and only if it intersects the object itself. 

If we are given a quasi--square with respect to $d = (d_x, d_y)$ and $d' = (d_x', d_y')$, we can transform it via the bijective linear transformation
\[ A = \lambda
\left( 
\left(
\begin{array}[h]{cc}
-d_x & d_x'\\
d_y & -d_y'
\end{array}
\right) \cdot
\left(
\begin{array}[h]{cc}
	\frac{l_d}{\| d \|} & 0\\
	0 & \frac{l_{d'}}{\| d' \|}
\end{array}
\right)
\right)^{-1}
\]
to yield an objects that is combinatorially equivalent to a unit square when only \ap lines are considered (here, $l_{d}, l_{d'}$ denote the lengths of the projections to the orthogonal complements of $d$ and $d'$, respectively).
The bounding box of $A\cdot o$ then is a square with side length $\lambda$. Also, the image of each line parallel to $d$ or $d'$ is axis--parallel As the transformation is bijective, we have 
\begin{proposition}\label{Proposition:BoundingBox} If $o$ is a quasi--square with respect to $d, d'$, for any  $\{d, d'\}$--line $l$ it holds that
\[ l\text{ intersects } o \iff A\cdot l\text{ intersects } A\cdot o \iff A\cdot l \text{ intersects } BB(A \cdot o). \] 
\end{proposition}
Thus, each instance with translates of $o$ and directions $\{d, d'\}$ is combinatorially equivalent to an instance with unit squares and \ap lines, and vice versa.

For connected objects that are not a point, also the constructions for the disjoint cases can easily be adapted: Thereto, we simply scale and rotate $o$ via a bijective linear transformation to fit inside $\rho$, the rectangle described in the previous section, such that it combinatorially is ``almost'' the same as $\rho$. Then placing such transforms of $o$ instead of $\rho$ in the set $\calR^*(G, k)$ and applying the inverse transformation again gives a set of disjoint translates of $o$ that can be stabbed by $6k$ arbitrary lines if and only if $\calR^*(G, k)$ can be stabbed by $6k$ arbitrary lines. We omit the technical details. See Figure \ref{fig:DisjointObjects}.
\begin{figure}[htbp]
	\centering
		\includegraphics[scale=.6]{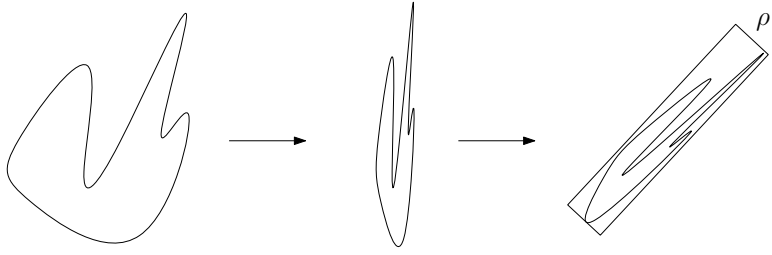}
		\caption{Transfomation of $o$}
	\label{fig:DisjointObjects}
\end{figure}
Using the remark at the end of the previous section, this proves Thm. \ref{thm:woneharddisjoint} (i) and (ii).

\subsection{Fixed Parameter Tractable Cases}\label{section:FPT}
In this section, we will consider several restricted versions of the above problems that are fixed parameter tractable. Here, all the objects are assumed to be closed, but again it is easy to modify the proofs to handle open objects as well.

\subsubsection{Stabbing disjoint \ap unit squares with \ap lines in the
  plane}\label{sssection:FPTDisjointUnitSquares}

To illustrate the idea, we first analyze the simplest case where the objects to be stabbed are disjoint \ap unit squares and the lines have to be axis--parallel.

Let $\calS$ be such a set of unit squares. Clearly it suffices to consider only lines that support the boundary of a square in $\calS$, so the total number of these \textit{relevant} lines is $2n + 2n$. Let $I(l)$ denote the set of squares in $\calS$ that are stabbed by $l$. A line $l$ is said to \emph{dominate} another line $l'$, if $I(l) \supseteq I(l')$.

The following data reduction rule is required for our algorithm to work:

\textbf{DR:} 
\emph{For all $\kappa > k+1$ squares with the same $x$--coordinates, delete
  all but $k+1$ of them, and the same for $\kappa > k+1$ squares that have
  the same $y$--coordinates.}

This rule is correct, i.\,e., the new set can be stabbed by $k$ lines if and only if the old one can: If there is a solution of size $k$ for the reduced set, then a solution of size $k$ for this set must contain a line that intersects all of those squares, for otherwise we would need at least $k+1$ lines. But any such line stabs all the deleted squares, too.\\

A set on which this data reduction rule is applied will be called a \textit{DR--set}. The following lemma states the main idea behind the algorithm:

\begin{lemma}\label{Lemma:BoundingLemma} 
  Let $l$ be a horizontal line that intersects $\kappa > k$ unit squares
  $I(l) = \{ \square_1(x_i, y_i) \mid 1 \leq i \leq \kappa \} \subseteq \mathcal
  S$. Then in order to stab the set $\mathcal S$ with $k$ lines, there
  has to be a horizontal line $l^*$ that intersects at least two squares
  from $I(l)$. Further, $l^*$ can be chosen from the set
  \[ B(I(l)) := \{ a_i \mid a_i\co y = y_i, 1 \leq i \leq \kappa \}
  \cup \{ b_i \mid b_i\co y = y_i + 1, 1 \leq i \leq \kappa \}. \]
\end{lemma}
\begin{proof}
There must be a line that intersects at least two of the squares because of the pigeonhole principle. This line cannot be vertical, as all of the squares are disjoint, i.\,e., no two of them can lie on both a common vertical and horizontal line.

We show that any such line is dominated by a line in $B(I(l))$. Let $I(l) = \{s_1, \dots, s_\kappa \}$, ordered from top to bottom, and let $l'$ be any line that intersects exactly the squares $s_i,\dots, s_j$ from $I(l)$ (and possibly others that are not in $I(l)$). Observe that always either $s_i = s_1$ or $s_j = s_\kappa$, as all squares have unit size. If both $s_i = s_1$ and $s_j = s_\kappa$, then $l'$ stabs all the squares at once and is thus dominated by either $a_1$ or $b_\kappa$. If $j < \kappa$ (the other case is symmetric) then no square that lies strictly above $l$, i.\,e. is not in $I(l)$ but intersected by $l'$, can have its upper side between $a_j$ and $l$, as $\dist(a_j, l) \leq 1$. Thus we have $I(l') \subseteq I(a_j)$. See Figure \ref{fig:BoundingLemma}.
\begin{figure}
	\centering
		\includegraphics[scale=.5]{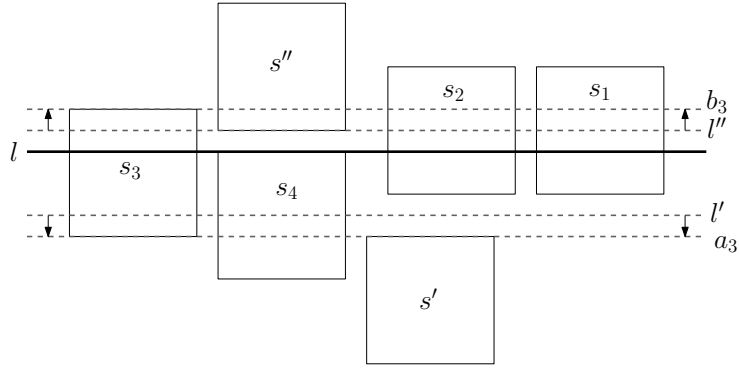}
	\caption{Here we have $I(l') \subseteq I(a_3)$ and $I(l'') \subseteq I(b_3)$}
	\label{fig:BoundingLemma}
\end{figure}
\end{proof}
For reasons of symmetry, an analogous lemma holds for the vertical lines as well. To prove that the algorithm is correct, we need another

\begin{lemma}\label{Lemma:IntermediateLemma} Let $\mathcal S$ be a DR--set. If there is an axis-parallel line $l$ with $|I(l)| > 2k + 1$, then there is also a line $l^*$ parallel to $l$ with $k+1 \leq |I(l^*)| \leq 2k+1$.
\end{lemma}
\begin{proof} Let $l$ be horizontal. Since $\mathcal S$ is a DR--set, the
  first relevant line above $l$ intersects at least $|I(l)| - (k + 1)$
  squares. In general, for two neighbouring relevant lines $l, l'$ we have that $\left||I(l)| - |I(l')|\right| \leq k + 1$. Further, the topmost relevant line stabs at most $k + 1$ squares,
  thus there must be a line $l^ *$ in between with $k+1 \leq |I(l^*)| \leq
  2k+1$. 
\end{proof}

We now come to describe the algorithm STAB($S$, $k$). In each call, it will find a line that stabs many ($k + 1$) but not too many ($2k + 2$) squares, if such a line exists, and otherwise use brute force.

\begin{algorithm}
\caption{STAB($S$, $k$)}
\label{alg:fpt1}
\begin{algorithmic}
\IF {$S = \emptyset$}	
      \STATE {``ACCEPT''}
\ELSIF {$k=0$} \RETURN 
\ENDIF
\STATE{apply DR}
\IF {there exists a line $l$ with $k + 1 \leq |I(l)| \leq 2k+1$}
  \FORALL {lines $l'$ from the set $B(I(l))$}
    \STATE {STAB($S - I(l')$, $k - 1$)}
  \ENDFOR
\ELSE
   \STATE {SOLVE($S, k$)} 
\ENDIF 
\end{algorithmic}
\end{algorithm}  


The SOLVE function simply counts if there are more than $k^2$ squares left and rejects in this case. Otherwise, it uses brute force by trying all $k$--subsets of the at most $4k^2$ relevant lines.
\begin{lemma} The algorithm accepts if and only if the set can be stabbed by $6k$ \ap lines.
\end{lemma}
\begin{proof}\ \\
``$\Rightarrow$'': Clearly, if the algorithm accepts, the set can be stabbed by $6k$ lines.\\
``$\Leftarrow$'': If there exists a line $l$ that intersects more than $k$ squares, then by Lemma \ref{Lemma:IntermediateLemma} there is a line $l^*$ with $k+1 \leq I(l^*) \leq 2k + 1$. By Lemma \ref{Lemma:BoundingLemma}, in any solution of size $k$ there must be a line that intersects at least two squares from $I(l^*)$. Further, any such line is dominated by a line in $B(I(l^*))$, and thus, if the set can be stabbed by $k$ lines, at least one of the branches ends up with an instance that can be stabbed by $k-1$ lines.

Otherwise, as mentioned above, we end up with an instance with at most $k^2$ squares left (otherwise we reject), and thus a solution can be found in fpt--time by the brute force algorithm.

\end{proof}
Thus, the algorithm is correct. To roughly determine the running time (a more sophisticated analysis will be given in the next section), observe that each call of the STAB function takes time $n^2$, if we simply calculate all the $I(l)$, and branches on at most $2(2k + 1)$ lines. Each of the branches ends up with a small instance which can be solved in $(4k^2)^k\cdot k^2$ steps, so the total running time is $\order\left((4k + 4)^{3k + 2}n^2\right)$. The algorithm runs in quadratic time for every fixed $k$ and thus is an fpt--algorithm. This completes the proof of Thm. \ref{thm:fpt} (i).

\subsubsection{Generalization}\label{sssec:FPTGeneralized}

A closer look on the above algorithm reveals that it really only depends on two properties of the set to be stabbed:
\begin{itemize}
	\item The squares are of unit size.
	\item A ``large'' set of squares that lie on a line in one direction cannot be intersected by ``few'' lines from another direction.
\end{itemize}
We will formalize these ideas and show how they can be generalized to work for different objects as well as for more than two directions. Thereto, let $\calD$ be a fixed set of directions. For a positive integer $c$, a set of objects is called $c$--shallow with respect to $\calD$, if for any two $\calD$--lines $l$, $l'$ it holds that
\[ |I(l) \cap I(l')| \leq c. \]
E.g., sets of disjoint unit squares with the property that each point lies
in at most $c$ squares are $c$--shallow with respect to \ap lines. Also,
for a \emph{fixed} rectangle $r$, sets of disjoint translates of $r$ are
$\order(1)$--shallow with respect to \ap lines. 
We show that the problem of stabbing $c$--shallow sets of objects that are translates of a connected object is fixed parameter tractable.

Let $\calD = \{d_1, \dots, d_r\}$, where the $d_i$ are lines, and $o$ be a connected object. Observe that it again suffices to consider the $2r \cdot n$ relevant lines that support the boundary of an object.
Given a $c$--shallow set of objects with respect to $\calD$, we first apply a generalized version of the above data reduction rule:\\
\textbf{DR': } 
\emph{Given $\kappa > ck + 1$ objects such that any line in direction $d_i$ intersects either all of them or none, delete all but $ck + 1$ of them.}\\
This data reduction rule is correct, as in the new set there must be a line that intersects $c+1$ of the squares at the same time, and any such line intersects all the $\kappa$ objects.

For two parallel lines $l\co ax + by = c$, $l'\co ax + by = c'$, we define \[ l < l' \co\iff c < c'. \] As the objects are closed, the functions 
\[ \max_{d}(s) := \max\{l \mid l \text{ is a $\{d\}$--line}, s \in I(l)\}\]
and 
\[ \min_{d}(s) := \min\{l \mid l \text{ is a $\{d\}$--line}, s \in I(l)\} \] 
are defined. Again, we can bound the number of lines to chose from:
\begin{lemma}\label{Lemma:BranchingNew} Let $l$ be a line in direction $d_i$ that intersects $\kappa > ck$ objects. Then in any solution of size $k$ there must be a line $l^*$ parallel to $l$ intersecting at least $c + 1$ of the objects. This line can be chosen from the set
\[ B(I(l)) := \left\{ \textnormal{max}_{d}(s) \mid s \in I(l) \right\} \cup \left\{ \textnormal{min}_{d}(s) \mid s \in I(l) \right\} \]
\end{lemma}
\begin{proof}
By rotating the entire set we can assume that $l$ is horizontal. Because of the pigeonhole principle there must be a line intersecting at least $c+1$ objects. No line not parallel to $l$ can intersect more than $c$ of the objects, for otherwise the set would not be $c$--shallow, thus in any solution of size $k$ there must be a line parallel to $l$. As the objects are all of the same size, by the same arguing as in Lemma \ref{Lemma:BoundingLemma}, any such line is dominated by a line from $B(I(l))$. 
\end{proof}

Also, similar to the above reasoning, if there exists a line for a DR'--set that intersects more than $2ck + 1$ objects, there must also be a parallel line $l^*$ with $ck + 1 \leq |I(l^*)| \leq 2ck + 1$. Thus, we can simply adapt the algorithm to the new bounds. We now apply, in each call of the STAB function, the new data recuction rule DR', and find a line $l$ with $ck + 1 \leq |I(l)| \leq 2ck + 1$, if it exists. Lemma \ref{Lemma:BranchingNew} ensures that it suffices to branch on the lines in $B(I(l))$. Thus, this algorithm accepts if and only if the set can be stabbed by $k$ $\calD$--lines.

\subsection{Running Time Analysis}
To analyze the running time, we split the algorithm into its three main steps and calculate them independently.
\paragraph{Data Reduction.}
The data reduction step can be done in time $O\left( r(n \log n) \right)$: First, we pick one of the $r$ directions and sort the objects according to this direction. Then we go through the array and delete all but $ck + 1$ out of each $\kappa > ck + 1$ have the same coordinates according to the direction (this takes only linear time). After that, we proceed with the next direction.
\paragraph{Call of the STAB--procedure.}
To find a line that stabs the desired number of objects, we again first pick one of the $r$ directions and sort the objects according to this direction. As they are connected, each of the objects implies two lines in each direction. For all of the $r\cdot 2n$ lines $l$ we then calculate whether $|I(l)| > 2ck + 1$. This requires $O(\log n + (2ck + 2))$ time by using binary search. As we have to do this at most $r$ times, it takes $O\left( r (n \log n + ck) \right)$ steps in total.
\paragraph{Solving the Problem Kernel.}
Let $m$ be the number of objects left. We reject the kernel if $m > ck^2$, as no line stabs more than $ck$ of them. 
Otherwise we can, instead of trying all of the $\approx m^k$ subsets of size $k$, use the following observation. Let $L(o)$ be the set of relevant lines through object $o$. By double--counting we get that
\[ 2r \cdot m \cdot ck \geq \sum_{\text{line }l}|I(l)| = \sum_{\text{object } o}|L(o)| \geq m \cdot \min_o|L(o)| \]
which yields $\min_o|L(o)| \leq 2rck$, and such an objects can be found in time $O(2r(ck)^2)$. Through any object there must be at least one line, so by branching on all the $2rck$ lines a solution is found if it exists. Thus, the kernel can be solved in time $O((2rck)^{k + 2})$.
\paragraph{Total Running Time.}
The algorithm branches on at most $2ck$ possibilities at most $k$ times, each step takes $O\left((2rck)^{k+2} + r \cdot n \log n\right)$, thus the total running time is
\[ O\left((2rck)^{2k+2} \cdot r n \log n) \right). \]
Thereby we have shown Thm. \ref{thm:fpt} (ii).


\section{Stabbing balls with one line}\label{balls_one_line}
We show that the problem of stabbing unit balls in $\Rd$ with a line is \wone-hard with respect to $d$ by an fpt-reduction from 
the \wone-complete $k$-independent set problem is general graphs~\cite{DF99}.
The reduction is based on a technique by Cabello et al.~\cite{CGKR08, cgkmr-gcfpt-09}. 
Given an undirected graph $G([n], E)$ we construct a set $\mathcal{B}$ of 
balls of equal radius $r$ in $\R^{2k}$ such that 
$\mathcal{B}$ can be stabbed by a line if and only if $G$ has an independent set of size $k$. First, we construct of a \emph{scaffolding}
set of balls that restricts the solutions to $n^k$ combinatorially
different solutions, which can be interpreted as potential
$k$-independent sets. Additional \emph{constraint} balls
will then encode the edges of the input graph.

The geometry of the construction
will be described as if
exact square roots and expressions of the form $\sin\frac{\pi}{n}$
were available.  To make the reduction suitable for the Turing machine
model, the data must be perturbed using fixed-precision roundings.  This
can be done with polynomially many bits in a way similar to the
rounding procedure followed in~\cite{CGKR08, cgkmr-gcfpt-09}. (We omit these technical details here).

\paragraph{Preliminaries.}
For every ball $B\in\mathcal{B}$ we will also have $-B\in\mathcal{B}$. This allows us to restrict our attention to lines 
through the origin: a line that stabs $\mathcal{B}$ can be translated so that it goes through 
the origin and still stabs $\mathcal{B}$. In this section, by a line we always mean a line through the origin. 
For a line $l$, let $\vec{l}$ be its unit direction vector. 
The notions of a point and vector will be used interchangeably.

It will be convenient to
view $\Reals^{2k}$ as the product of $k$ orthogonal planes $E_1,\dots, E_k$,
where each $E_i$ has coordinate axes $X_i,Y_i$. The origin is denoted by $o$.
The coordinates of a point $p\in\R^{2k}$ are denoted by $\left(x_1(p), y_1(p),\ldots ,x_k(p), y_k(p)\right)$.
We denote by $C_i$ the unit circle on $E_i$ centered at $o$.

\subsection{Scaffolding ball set}
\label{scaff}
For each plane $E_i$, we define $2n$ $2k$-dimensional balls, 
whose centers $c_{i1},\ldots ,c_{i2n}$ are regularly spaced on the circle $C_i$.
Let $c_{iu}\in E_i$ be the center of the ball $B_{iu}$, $u\in[2n]$, with  
$$
x_i(c_{iu})=\cos(u-1)\tfrac{\pi}{n},\; y_i(c_{iu})=\sin(u-1)\tfrac{\pi}{n}.
$$
We define the scaffolding ball set $\mathcal{B}^0=\{B_{iu},\, i=1,\ldots,k \;\mathrm{and}\; u=1,\ldots,2n\}$. 
We have $|\mathcal{B}^0|=2nk$. All balls in $\mathcal{B}^0$ will have the same radius $r<1$, to be defined later.

Two antipodal balls $B$, $-B$ are stabbed by the same set of lines.
A line $l$ stabs a ball $B$ of radius $r$ and center $c$ if and only if $(c\cdot \vec{l})^2\geq \lVert c\rVert^2-r^2$. 
Thus, $l$ stabs $\mathcal{B}^0$ if and only if it satisfies the following system of $nk$ inequalities:
\begin{equation}
\nonumber 
(c_{iu}\cdot \vec{l})^2\geq \lVert c_{iu}\rVert^2-r^2=1-r^2, 
\;\;\text{for}\;\;
i=1,\ldots,k
\;\;\text{and}\;\;
u=1,\ldots,n. 
\end{equation}

Consider the inequality asserting that $l$ stabs $B_{iu}$.
Geometrically, it amounts to saying that
the projection $\vec{l}_i$ of $\vec{l}$ on the plane $E_i$ lies in one of the half-planes
\begin{equation*}
H_{iu}^+=\{p\in E_i | c_{iu}\cdot p \geq \sqrt{\lVert c_{iu}\rVert^2-r^2}\} 
\;\;\mathrm{or}\;\; 
H_{iu}^-=\{p\in E_i | c_{iu}\cdot p \leq -\sqrt{\lVert c_{iu}\rVert^2-r^2}\}.
\end{equation*}
Consider the situation on a plane $E_i$.
Looking at all half-planes $H_{i1}^+, H_{i1}^-,\ldots,H_{in}^+, H_{in}^-$, we see that 
$l$ stabs all balls $B_{iu}$ (centered on $E_i$) if and and only if $\vec{l}_i$ lies in one of the $2n$ wedges
$\pm (H_{i1}^- \cap H_{i2}^+),\ldots, \pm (H_{i(n-1)}^- \cap H_{in}^+), \pm (H_{i1}^- \cap H_{in}^-)$; 
see Fig.~\ref{wedges}. 
\begin{figure}
  \centering
	 \includegraphics[width=7.5cm]{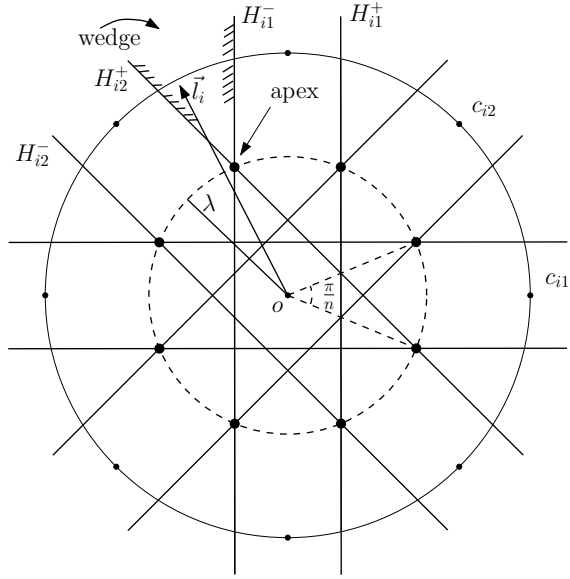}
	 \caption{Centers of the balls and their respective half-planes and wedges on a plane $E_i$, for $n=4$.}
	 \label{wedges}
\end{figure}
The apices of the wedges are regularly spaced on a circle of radius 
$\lambda=\sqrt{2(1-r^2)/(1-\cos\frac{\pi}{n})}$, 
and define the set 
\begin{equation*}
A_i=\{\pm \left(\lambda \cos(2u-1)\tfrac{\pi}{2n}, \lambda \sin(2u-1)\tfrac{\pi}{2n}\right)\in E_i,\, u=1,\ldots,n\}.
\end{equation*}
For $l$ to stab all balls $B_{iu}$, we must have that $\lVert \vec{l}_i\rVert\geq\lambda$.
We choose $r=\sqrt{1-(1-\cos\frac{\pi}{n})/(2k)}$ in order to obtain $\lambda=1/\sqrt{k}$.  

Since the above hold for every plane $E_i$, and since $\vec{l}\in\R^{2k}$ is a unit vector, 
we have
$$
1=\lVert l\rVert^2=\lVert l_1\rVert^2 +\cdots +\lVert l_k\rVert^2 \geq k\lambda^2 = 1.
$$
 Hence, equality holds throughout, which implies that $\lVert\vec{l}_i\rVert=1/\sqrt{k}$, for every $i\in\{1,\ldots,k\}$. Hence, for 
line $l$ to stab all balls in $\mathcal{B}^0$, every projection $\vec{l}_i$ must be one 
of the $2n$ apices in $A_i$.  Each projection $\vec{l}_i$ can be chosen independently. There are $2n$ choices, but since $\vec{l}$ and 
$-\vec{l}$ correspond to the same line, the total number of lines that stab $\mathcal{B}^0$ is $n^k 2^{k-1}$. 

For a tuple $(u_1,\ldots,u_k)\in[2n]^k$, we will denote by $l(u_1,\ldots, u_k)$ the stabbing line with direction vector
\begin{equation*}
\frac{1}{\sqrt{k}}\left(\cos(2u_1-1)\tfrac{\pi}{2n}, \sin(2u_1-1)\tfrac{\pi}{2n}, \ldots, 
\cos(2u_k-1)\tfrac{\pi}{2n}, \sin(2u_k-1)\tfrac{\pi}{2n}\right).
\end{equation*}
Two lines $l(u_1,u_2,...,u_k)$
and $l(v_1,v_2,...,v_k)$ are said to be equivalent
if $u_i \equiv v_i\pmod n$, for all $i$. 
This relation defines $n^k$ equivalence classes $L(u_1,\ldots,u_k)$, with $(u_1,\ldots,u_k)\in[n]^k$, 
where each class consists of $2^{k-1}$ lines. 

From the discussion above, it is clear that there is a bijection between the possible equivalence classes of lines that 
stab $\mathcal{B}^0$ and $[n]^k$. 
\subsection{Constraint balls}

We continue the construction of the ball set $\mathcal{B}$ by showing how
to encode the structure of~$G$.  For each pair of distinct indices $i\neq
j$ ($1\leq i,j \leq k$) and for each pair of (possibly equal) vertices
$u,v\in[n]$, we define a \emph{constraint set} $\mathcal{B}_{ij}^{uv}$ of balls
with the property that (all lines in) all classes $L(u_1,\ldots,u_k)$ stab
$\mathcal{B}_{ij}^{uv}$ except those with $u_i=u$ and $u_j=v$. The centers
of the balls in $\mathcal{B}_{ij}^{uv}$ lie in the $4$-space $E_i\times
E_j$. Observe that all lines in a particular class $L(u_1,\ldots,u_k)$ project onto
only two lines on $E_i\times E_j$.
We use a ball $B_{ij}^{uv}$ (to be defined shortly) of radius~$r$ that is stabbed by \emph{all} lines $l(u_1,\ldots,u_k)$ except those 
with $u_i=u$ and $u_j=v$.
Similarly, we use a ball $B_{ij}^{u\bar v}$ that is stabbed by \emph{all} lines $l(u_1,\ldots,u_k)$ except those 
with $u_i=u$ and $u_j=\bar v$, where $\bar v=v+n$. 
Our constraint set consists then of the four balls
$$\mathcal{B}_{ij}^{uv}=\{\pm B_{ij}^{uv}, \pm B_{ij}^{u\bar v}\}.$$

We describe now the placement of a ball $B_{ij}^{uv}$. Consider a line $l=l(u_1,\ldots,u_k)$ with 
$u_i=u$ and $u_j=v$. The center $c_{ij}^{uv}$ of $B_{ij}^{uv}$ will lie on a line ${z}\in E_i\times E_j$ that is
orthogonal to~$\vec{l}$, but not orthogonal to any line $l(u_1,\ldots,u_k)$ with 
$u_i\neq u$ or $u_j\neq v$. We choose the direction
$\vec{z}$ of $z$ as follows: 
$$x_i(\vec{z})=\mu(\cos\theta_i-3n\sin\theta_i),\, 
y_i(\vec{z})=\mu(\sin\theta_i+3n\cos\theta_i),$$ 
$$x_j(\vec{z})=\mu(-\cos\theta_j-6n^2\sin\theta_j),\, 
y_j(\vec{z})=\mu(-\sin\theta_j+6n^2\cos\theta_j),$$
where $\theta_i=(2u-1)\frac{\pi}{2n}$, $\theta_j=(2u-1)\frac{\pi}{2n}$, and $\mu=1/(9n^2+36n^4+2)$. 
It is straightforward to check that $\vec{l}\cdot\vec{z}=0$.

Let $\omega$ 
be the angle between $\vec{l'}$ and $\vec{z}$. 
We have the following lemma: 
\begin{lemma}
\label{omega_bound}
For any line $l'=l(u_1,\ldots,u_k)$, with $u_i\neq u$ or $u_j\neq v$ the angle $\omega$ 
between $\vec{l'}$ and $\vec{z}$ satisfies $|\cos\omega|>\frac{\mu}{\sqrt{k}}$.
\end{lemma}
\begin{proof}
Without loss of generality we consider a fixed direction $\vec{z}$ where $\theta_i=\theta_j=\frac{\pi}{2n}$ (i.\,e., $u=v=1$).
Consider $\vec{l'}$ with $x_i(\vec{l'})=\cos\theta$, $y_i(\vec{l'})=\sin\theta$, $x_j(\vec{l'})=\cos\phi$, and 
$y_j(\vec{l'})=\sin\phi$, where $\theta=(2u_i-1)\frac{\pi}{2n}$ and $\phi=(2u_j-1)\frac{\pi}{2n}$, with $(u_i,u_j)\neq (1,1)$ and 
$(u_i,u_j)\neq (n+1, n+1)$. After straightforward calculations we have that
$|\cos\omega|=|\vec{l'}\cdot\vec{z}|=\frac{\mu}{\sqrt{k}}|\alpha|$, where 
$$\alpha=\cos(u_i-1)\tfrac{\pi}{n} + 3n\sin(u_i-1)\tfrac{\pi}{n} - \cos(u_j-1)\tfrac{\pi}{n} + 6n^2\sin(u_j-1)\tfrac{\pi}{n}.$$
We will show that $|\alpha|>1$. We will use the inequality:
$$|\sin(u_i-1)\tfrac{\pi}{n}|\geq|\sin\tfrac{\pi}{n}|>\tfrac{1}{n},$$
which holds for all $1\leq u_i\leq 2n$, with $u_i\neq 1$, $u_i\neq n+1$, and $n\geq 4$.
We examine the following cases:

(i) $u_j\neq 1$ and $u_j\neq n+1$. Then $u_i$ can take any value.
We have 
\begin{align*}
|\alpha| &\geq
\left||6n^2\sin(u_j-1)\frac{\pi}{n}|-|\cos(u_j-1)\frac{\pi}{n} - \cos(u_i-1)\frac{\pi}{n} - 3n\sin(u_i-1)\frac{\pi}{n}|\right|
\\&
>|6n^2\cdot \tfrac{1}{n} - |2+3n||
\\&
=3n-2 > 1.
\end{align*}

(ii) $u_j=1$. Then $u_i\neq 1$. If also $u_i\neq n+1$, we have
\begin{align*}
|\alpha| &\geq
|0-1+3n\sin(u_i-1)\frac{\pi}{n} + \cos(u_i-1)\frac{\pi}{n}|
\\&
>|-1+3n\cdot \frac{1}{n}-1| = 1.
\end{align*} 
If $u_i=n+1$, then $|\alpha|=2$.

(iii) $u_j=n+1$. Then $u_i\neq n+1$. The two cases where $u_i\neq 1$ or $u_i=1$ are dealt with similarly to the previous case.
\end{proof}
This lower bound on $|\cos\omega|$ helps us 
place $B_{ij}^{uv}$ sufficiently close to the origin so that it is still intersected by $l'$, i.\,e., 
$\vec{l'}$ lies in one of the half-spaces 
$c_{ij}^{uv}\cdot p\geq \sqrt{\lVert c_{ij}^{uv}\rVert^2-r^2}$ or $c_{ij}^{uv}\cdot p\leq -\sqrt{\lVert c_{ij}^{uv}\rVert^2-r^2}$, 
$p\in\R^{2k}$.

We claim that any point $c_{ij}^{uv}$ on $z$ with $r<\lVert c_{ij}^{uv}\rVert<\sqrt{\frac{k}{k-\mu^2}}r$ will do.
For any position of $c_{ij}^{uv}$ on $z$ with $\lVert c_{ij}^{uv}\rVert>r$, we have 
$(c_{ij}^{uv}\cdot \vec{l})^2=0<\lVert c_{ij}^{uv}\rVert^2-r^2$, 
i.\,e., $l$ does not stab $B_{ij}^{uv}$. 
On the other hand, as argued above we need that $|c_{ij}^{uv}\cdot \vec{l'}|\geq \sqrt{\lVert c_{ij}^{uv}\rVert^2-r^2}$.
Since $c_{ij}^{uv}\cdot \vec{l'}=\cos\omega\cdot \lVert c_{ij}^{uv}\rVert$, 
we have the condition $|\cos\omega|\geq\sqrt{1-\frac{r^2}{\lVert c_{ij}^{uv}\rVert^2}}$. By Lemma~\ref{omega_bound} we know that 
$|\cos\omega|>\frac{\mu}{\sqrt{k}}$, hence by choosing $\lVert c_{ij}^{uv}\rVert$ 
so that $\frac{\mu}{\sqrt{k}}>\sqrt{1-\frac{r^2}{\lVert c_{ij}^{uv}\rVert^2}}$ 
we are done.
\paragraph{Reduction.}
Similarly to~\cite{CGKR08}, the structure of the input graph $G([n], E)$ can now be represented as follows.
We add to $\mathcal{B}^0$ the $4n\binom k 2$ balls in 
$\mathcal{B}_V=\bigcup \mathcal{B}_{ij}^{uu},\, 1\leq u\leq n, \, 1\leq i<j\leq k$, to ensure that 
all components $u_i$ in a solution (class of lines $L(u_1,\ldots,u_k)$) are distinct. 
For each edge $uv\in E$ we also add the balls in $k(k-1)$ sets $\mathcal{B}_{ij}^{uv}$, with $i\neq j$.
This ensures that the remaining classes of lines $L(u_1,\ldots,u_k)$ represent independent sets of size $k$.
In total, the edges are represented by the $4k(k-1)|E|$ balls in 
$\mathcal{B}_E=\bigcup \mathcal{B}_{ij}^{uv},\, uv\in E,\, 1\leq i,j\leq k,\, i\neq j$.
The final set $\mathcal{B}=\mathcal{B}^0\cup\mathcal{B}_V\cup\mathcal{B}_E$ has $2nk+4\binom k 2 (n+2|E|)$ balls.

As noted in above, there is a bijection between the possible equivalence classes of lines $L(u_1,\ldots,u_k)$ that 
stab $\mathcal{B}$ and the tuples $(u_1,\ldots,u_k)\in [n]^k$. The constraint sets of balls exclude tuples with two equal indices 
$u_i=u_j$ or with indices $u_i$, $u_j$ when $u_iu_j\in E$, thus, the classes of lines that stab $B$ represent exactly the independent 
sets of $G$. Thus, we have the following:
\begin{lemma}
Set $\mathcal{B}$ can be stabbed by a line if an only if $G$ has an independent set of size $k$.
\end{lemma}
From this lemma and since this is an fpt-reduction, Theorem~\ref{thm:balls} follows.


\bibliographystyle{abbrv}
\bibliography{stabbing}

\end{document}